\documentclass{article}

\usepackage{authblk}
\usepackage{amsthm,amsmath,amssymb}
\usepackage{stmaryrd,color,bm,nicefrac,url}
\usepackage{color}    
\usepackage{here}
\usepackage{times}
\usepackage{graphicx, subfigure,stmaryrd,enumerate}
\usepackage{proof}

\usepackage[inline]{enumitem}

\usepackage{hyperref}
\usepackage{tikz}

\newtheorem{theorem}{Theorem}[section]
\newtheorem{lemma}[theorem]{Lemma}
\newtheorem{definition}[theorem]{Definition}
\newtheorem{proposition}[theorem]{Proposition}
\newtheorem{corollary}[theorem]{Corollary}
\newtheorem{remark}[theorem]{Remark}
\newtheorem{exam}[theorem]{Example}

\newcommand{\posres}[2]{\nicefrac{#1}{#2}}
\newcommand{\fullproof}[1]{}
\newcommand{\shortproof}[1]{#1}
\newcommand{\logpersax}{{\sf ITL}^{\sf FS}}
\newcommand{\lanclass}{{\mathcal L}^C}

\newcommand{\iiff}{\leftrightarrow}
\newcommand{\ineg}{\neg}

\newcommand{\icltl}{{\sf ITL}^{\sf c}_\ps }

\newcommand{\cl }{\mathcal }
\newcommand{\fr }{\mathfrak }

\newcommand{\nx}{{\circ}}
\newcommand{\nec}{\Box}
\newcommand{\ps}{\Diamond}
\newcommand{\val}[1]{\lb #1 \rb}

\def\eqdef{\stackrel{\rm def}{=}}
\usepackage{color}
\usepackage{multicol}
\usepackage{bm}

\usepackage{ebproof}
\usepackage{cancel}

\newcommand{\forrefs}[1]{}

\newcommand{\sqsubT}{\sqsubseteq _T}
\newcommand{\simrel}{\mathrel E}
\newcommand{\simu}{\rightharpoonup}
\newcommand{\remc}[1]{\ominus #1}
\newcommand{\CIcon }{\Lambda}
\newcommand{\ptype}{\type\infty}
\newcommand{\ptypel}[1]{T^\infty_{#1}}
\newcommand{\Sim}[1]{{\rm Sim}({#1})}
\newcommand{\CMod}{\model_\CIcon }
\newcommand{\model}{{\mathfrak M}}
\newcommand{\tnext}{\circ}
\newcommand{\peq}{\preccurlyeq}

\newcommand{\ubox}{\Box}
\newcommand{\diam}{\Diamond}
\newcommand{\itlc}{{\sf ITL}^{\sf c}}
\newcommand{\itle}{{\sf ITL}^{\sf e}}

\newcommand{\ignore}[1]{}
\newcommand{\peqT}{\peq_{T}}
\newcommand{\subT}{\subseteq_{T}}
\newcommand{\ST}{\mathrel S_T}

\newcommand{{\logbasic}}{{\sf ITL}^0_\diam}

\newcommand{\type}[1]{T_{ #1 }}

\def\nb{\blacksquare}

\newcommand{\lanfull}{{\cl L}_\ast}

\newcommand{\landi}{\cl L_\diam}
\newcommand{\landif}{\cl L_{\diam\forall}}

\newcommand{\unp}[1]{J _{#1}}

\newcommand{\cval}[1]{\lb #1 \rb^C}

\newcommand\cqm[1]{\fr J _{ #1 }}
\newcommand{\irr}[1]{\fr I_{ #1 }}
\newcommand{\remove}[2]{#1 \setminus #2}

\def\lb{\left\llbracket}
\def\rb{\right\rrbracket}
\def\<{\left (}

\def\>{\right )}
\def\({\left (}
\def\){\right )}

\def\cbra{\left \{}
\def\cket{\right \}}

\def\eqdef{\stackrel{\rm def}{=}}

\def\lb{\left\llbracket}
\def\rb{\right\rrbracket}
\def\<{\left (}

\def\>{\right )}

\def\cbra{\left \{}
\def\cket{\right \}}

\newcommand{\imp}{\mathop \to}
\newcommand{\seq}{\succcurlyeq}

\newcommand{\defect}[1]{\partial {#1}}

\def\M{{\mathcal M}}

    \newcommand{\david}[1]{}
     
\makeatletter
\newcommand{\mylabel}[2]{#2\def\@currentlabel{#2}\label{#1}}
\makeatother

\begin{document}

\title{Complete Intuitionistic Temporal Logics in Topological Dynamics}
\author[1]{Joseph Boudou \footnote{\href{mailto:joseph.boudou@matabio.net}{\tt joseph.boudou@matabio.net}}}
\author[2]{Mart\'in Di\'eguez \footnote{\href{mailto:dieguez@enib.fr}{\tt dieguez@enib.fr}}}
\author[3]{David Fern\'andez-Duque \footnote{\href{mailto:david.fernandezduque@ugent.be}{\tt david.fernandezduque@ugent.be}}}
\affil[1]{IRIT, France}
\affil[2]{CERV, ENIB,LAB-STICC. Brest, France}
\affil[3]{Department of Mathematics, Ghent University. Gent, Belgium}

{\maketitle}

\begin{abstract}
The language of linear temporal logic can be interpreted over the class of dynamic topological systems, giving rise to the intuitionistic temporal logic $\itlc_{\ps\forall}$, recently shown to be decidable by Fern\'andez-Duque.
In this article we axiomatize this logic, some fragments, and prove completeness for several familiar spaces.
\end{abstract}

\section{Introduction}

The dynamic topological logic project originated in the work of Artemov et al.~\cite{arte}, who suggested that modal logic may be used to reason about dynamic topological systems $(X,f)$ using the `interior' $\blacksquare$ modality (in the sense of Tarski \cite{tarski}) and the `next' $\nx$ modality to reason about the action of $f$.
Kremer and Mints \cite{kmints} later observed that `henceforh' $\nec$ could be used to model the asymptotic behavior of $f$, allowing one to represent phenomena such as the Poincar\'e recurrence theorem \cite{Poincare}. The resulting tri-modal system was called {\em dynamic topological logic} ($\sf DTL$), and it was studied for some time with expectations that it may be applicable in e.g.~automated theorem proving. However, interest in $\sf DTL$ waned after Konev et al.~\cite{konev} showed that the validity problem for $\sf DTL$ formulas is undecidable.

In unpublished work, Kremer \cite{KremerIntuitionistic} also suggested an intuitionistic version of $\sf DTL$, based on `next' and `henceforth'; no interior modality is required in this setting, as the topology is reflected in the semantics for implication. Years later, Fern\'andez-Duque \cite{FernandezITLc} showed that a mild variant based on `next', `eventually' and a universal modality was decidable over the class of all dynamical systems, and Boudou et al.~\cite{BoudouCSL} showed that a logic with `next', `eventually' and `henceforth' was decidable over the class of all dynamical systems based on a poset (see \S \ref{SecBasic}).
These results have led us to recast the dynamic topological logic project in terms of intuitionistic temporal logics.

These decidability results were proven using model-theoretic techniques, with little being known regarding axiomatic systems. Thus our goal in this paper is to develop deductive calculi for these logics. 

\subsection{State-of-the-art}

There are several (poly)modal logics which may be used to model time, and some have already been studied in an intuitionistic setting, e.g.~tense logics by Davoren \cite{Davoren2009} and propositional dynamic logic with iteration by Nishimura \cite{NishimuraConstructivePDL}.
Here we are specifically concerned with intuitionistic analogues of discrete-time linear temporal logic.
Versions of such a logic in finite time have been studied by Kojima and Igarashi \cite{KojimaNext} and Kamide and Wansing \cite{KamideBounded}.
Nevertheless, logics over infinite time have proven to be rather difficult to understand, in no small part due to their similarity to intuitionistic modal logics such as $\sf IS4$, whose decidablitiy is still an open problem~\cite{Simpson94}.

Balbiani and the authors have made some advances in this direction, showing that the intermediate logic of temporal here-and-there is decidable and enjoys a natural axiomatization \cite{BalbianiDieguezJelia} and identifying several conservative temporal extensions of intuitionistic logic, interpreted over dynamic topological systems \cite{FernandezITLc} or what we call {\em expanding posets} \cite{BoudouCSL}. These logics are based on the temporal language with $\tnext$ (`next'), $\diam$ (`eventually'), $\ubox$ (`henceforth') and the universal modality $\forall$. Note that unlike in the classical case, $\diam$ and $\ubox$ are not inter-definable \cite{IMLA}.

Fern\'andez-Duque \cite{FernandezITLc} has shown that a formula in the $\ubox$-free language $\landif$ is valid if and only if it is valid over a class of suitably defined {\em quasimodels.}
With this he showed that the validity problem for $\landif$ over the class of dynamical systems is decidable.
It is also shown in \cite{FernandezITLc} that $\landif$ is expressive enough to capture many of the recurrence phenomena that inspired interest in $\sf DTL$, such as Poincar\'e recurrence and minimality \cite{Bir,FernandezMinimal}.
As we will show, quasimodels can also be used to prove the completeness of a natural deductive calculus ${\sf ITL}^0_{\ps\forall}$ for this language.

On the other hand, we do not yet have a useful notion of quasimodel in the presence of $\ubox$, and Kremer \cite{KremerIntuitionistic} has shown that some key axioms of classical $\sf LTL$ (including $\ubox \varphi \to \nx \ubox \varphi$) are not valid for his topological semantics. At this point, it is unclear which weaker principles should replace them.
For these reasons, in this manuscript we restrict our attention to $\ubox$-free fragments of the temporal language.

\subsection{Our main result}

The main goal of this article is to prove that ${\sf LTL}^0_{\ps \forall}$ is complete for the class of dynamical systems (Theorem \ref{theocomp}).
The completeness proof follows the general scheme of that for linear temporal logic \cite{temporal}: a set of `local states', which we will call {\em moments,} is defined, where a moment is a representation of a potential point in a model (or, in our case, a quasimodel). To each moment $w$ one then assigns a characteristic formula $\chi(w)$ in such a way that $\chi(w)$ is consistent if and only if $w$ can be included in a model, from which completeness can readily be deduced.

In the $\sf LTL$ setting, a moment is simply a maximal consistent subset of a suitable finite set $\Sigma$ of formulas.
For us a moment is instead a finite labelled tree, and the formula $\chi(w)$ must characterize $w$ up to {\em simulation;} for this reason we will henceforth write $\Sim{w}$ instead of $\chi(w)$.
The required formulas $\Sim{w}$ can readily be constructed in our language (Proposition \ref{propSimForm}).

Note that it is {\em failure} of $\Sim{w}$ that characterizes the property of simulating $w$, hence the {\em possible} states will be those moments $w$ such that $\Sim{w}$ is unprovable.
The set of possible moments will form a quasimodel falsifying a given unprovable formula $\varphi$ (Corollary \ref{laststretch}), from which it follows that such a $\varphi$ is falsified on some model as well (Theorem \ref{TheoITLc}). Thus any unprovable formula is falsifiable, and Theorem \ref{theocomp} follows.

Our proof will be presented in such a way that completeness for the sub-logics ${\sf ITL}^0_\nx$ (whose only modality is $\nx$) and ${\sf ITL}^0_\ps$ (whith $\nx$ and $\ps$) are obtained as partial results.
Once we have established our main completeness theorem, we will consider special classes of dynamical systems for which our logics are also complete. In summary, we obtain the following results.

\begin{enumerate}

\item ${\sf ITL}^0_\nx$ and ${\sf ITL}^0_\ps$ are sound and complete for the class of expanding posets, for dynamical systems based on $\mathbb R^n$ for any fixed $n\geq 2$, and for the Cantor space.

\item The logics $\logpersax_\nx$ and $\logpersax_{\nx\forall}$ are sound and complete for the class of {\em persistent} posets.

\item The logic ${\sf ITL}^0_{\ps\forall}$ is sound and complete for both the class of all dynamical systems and of all dynamical systems based on $\mathbb Q$.

\end{enumerate}
In contrast, we will also show that ${\sf ITL}^0_\nx$ is incomplete for $\mathbb R$.

\subsection*{Layout} Section \ref{SecTopre} reviews some basic notions regarding partial orders and topology and Section \ref{SecBasic} introduces the syntax and semantics of $\icltl$. Section \ref{SecNDQ} then discusses labelled structures, which generalize both models and quasimodels.
Section \ref{secCanMod} discusses the canonical model, which properly speaking is only a deterministic weak quasimodel but is sufficient to establish completeness results for logics over the language $\mathcal L_\nx$.
Section \ref{SecSim} reviews simulations and dynamic simulations, including their definability in the intuitionistic language.
Section \ref{seccan} constructs the initial quasimodel and establishes its basic properties, but the fact that it is actually a quasimodel is proven only in Section \ref{secOmSens} where it is shown that the quasimodel is {\em $\omega$-sensible,} i.e.~it satisfies the required condition to interpret $\diam$. The completeness of ${\logbasic}$ follows from this.
In Section \ref{secUniversal} we extend our axiomatization for languages with the universal modality and prove that the logic ${\sf ITL}^0_{\ps\forall}$ is complete.
In Section \ref{SecCons} we prove via an unwinding construction that any $\mathcal L_\ps$-formula falsifiable on a quasimodel is also falsifiable on an expanding poset, from which we conclude that $\logbasic$ is complete for this class of systems, and in Section \ref{secMetric} show how to adapt results for $\sf DTL$ to our setting, obtaining completeness results for $\mathbb Q$, $\mathbb R^n$ and the Cantor space. Finally, Section \ref{secConc} provides some concluding remarks.

\section{Posets and Topology}\label{SecTopre}

We assume familiarity with topological spaces and related concepts; the necessary background may be found in a text such as \cite{munkres2000}.
Topological spaces will typically be denoted $\fr X = (|\fr X|,\mathcal T_\fr X)$, i.e.~$|\fr X|$ is the set of points of $\fr X$ and $\mathcal T_\fr X$ the family of open sets; we will generally adopt the convention of denoting the domain of a stucture $\fr S$ by $ | \fr S |$. We denote the interior of $A\subseteq |\fr X|$ by $A^\circ$ and its closure by $\overline A$. 

We will also work with posets and it will be convenient to view them as a special case of topological spaces.
As usual, a {\em poset} is a pair $\mathfrak A=(|\mathfrak A|,\peq_\mathfrak A)$, where $|\mathfrak A|$ is any set and ${\peq}_\mathfrak A \subseteq |\mathfrak A|\times |\mathfrak A|$ is a reflexive, transitive, antisymmetric relation. We write $\peq$ instead of $\peq_\mathfrak A $ when this does not lead to confusion, and write $a\prec b$ for $a\peq b$ but $b\not \peq a$.
 If $B\subseteq |\fr A|$, $\fr A\upharpoonright B$ is the structure obtained by restricting each component of $\fr A$ to $B$, so that $\fr A \upharpoonright B = (B, {\peq_\mathfrak A} \cap (B\times B))$. Similar conventions apply to other classes of structures, so that for example if $\mathfrak X$ is a topological space and $Y\subseteq |\fr X|$ then $\fr X \upharpoonright Y$ is $Y$ equipped with the subspace topology.

If $\mathfrak W$ is a poset, consider the topology $\mathcal U_\peq$ on $|\mathfrak W|$ given by setting $U\subseteq|\mathfrak W|$ to be open if and only if, whenever $w\in U$ and $v\seq w$, we have that $v\in U$. We call $\mathcal U_\peq$ the {\em up-set topology of $\peq$.}
Topological spaces of this form are {\em Aleksandroff spaces} \cite{alek}, which are fundamental for our completeness proof. If $\mathfrak W,\fr V$ are preorders then it is not hard to check that $g\colon |\mathfrak W|\to |\mathfrak V|$ is continuous with respect to the up-set topologies on $\mathfrak W,\fr V$ if and only if $v\peq_\fr W w$ implies that $g(v)\peq_\fr V g(w)$.

\section{Syntax and Semantics}\label{SecBasic}

Fix a countably infinite set $\mathbb P$ of propositional variables. The {\em full (intuitionistic temporal) language} $\lanfull$ is defined by the grammar (in Backus-Naur form)
\[\varphi,\psi := \ \  \bot \  | \   p  \ |  \ \varphi\wedge\psi \  |  \ \varphi\vee\psi  \ |  \ \varphi\imp \psi  \ |  \ \nx\varphi \  | \  \ps\varphi \  |  \ \nec\varphi  \ |  \  \forall \varphi, \]
where $p\in \mathbb P$. Here, $\nx$ is read as `next', $\ps$ as `eventually', $\nec$ as `henceforth' and $\forall$ as `everywhere'; note that this is a universal modality and not a quantifier. We also use $\ineg\varphi$ as a shorthand for $\varphi\imp \bot$ and $\varphi\iiff \psi$ as a shorthand for $(\varphi\imp \psi) \wedge (\psi\imp\varphi)$; we remark that the exitential modality is definable by $\exists \varphi := \neg \forall \neg \varphi$~\cite{FernandezITLc}. We denote the set of subformulas of $\varphi\in\lanfull$ by ${\mathrm{sub}}(\varphi)$.

The sublanguage of $\lanfull$ which only allows modalities in $M\subseteq \{\ps,\nec,\forall\}$ is denoted $\mathcal L_ M$, where we omit brackets and commas when writing $M$. For purposes of this article, a {\em temporal language} is any language of this form. Note that $\mathcal L_ M$ always contains Booleans, $\nx$ and implication, even though they will not be listed in $M$. We write $\mathcal L_\nx$ instead of $\mathcal L_\varnothing$ (i.e., the language whose only modality is $\nx$).

Next we define our semantics, based on dynamical systems.

\begin{definition}\label{DefSem}
A {\em dynamical (topological) system} is a triple $\mathfrak X=(|\mathfrak X|,\mathcal{T}_\mathfrak X,f_\mathfrak X)$
where $(|\mathfrak X|,\mathcal{T}_\mathfrak X)$ is a topological space and $f_\fr X:|\mathfrak X|\to |\mathfrak X|$ is a continuous function.
A {\em valuation on $\mathfrak X$} is a function $\lb\cdot\rb\colon\lanfull \to \mathcal T_\mathfrak X$ such that
\[
\begin{array}{rclrcl}
\lb\bot\rb&=&\varnothing &
\lb\varphi\wedge\psi\rb &=&\lb\varphi\rb\cap \lb\psi\rb\\
\lb\varphi\vee\psi\rb &=&\lb\varphi\rb\cup \lb\psi\rb&
\lb\varphi\imp\psi\rb &=&\big ( (|\mathfrak X|\setminus\lb\varphi\rb)\cup \lb\psi\rb\big )^\circ\\
\val{\nx\varphi}&=&f^{-1}_\mathfrak X\val\varphi&
\val{\ps\varphi}&=& \bigcup_{n<\omega}f^{-n}_\mathfrak X\val\varphi\\
\val{\nec\varphi}&=&\big (\bigcap_{n<\omega}f^{-n}\val\varphi\big)^\circ&
\val{\forall\varphi}&=&
|\mathfrak X|\text{ if $\val\varphi=|\mathfrak X|$,
$\varnothing$ otherwise.}
\end{array}
\]
A dynamical system $\mathfrak X$ equipped with a valuation $\lb\cdot\rb_\mathfrak X$ is a {\em (dynamical topological) model}.
\end{definition}
Validity is then defined in the usual way:

\begin{definition}
Given a model $\mathfrak X$ and a formula $\varphi\in \lanfull$, we say that $\varphi$ is {\em valid} on $\mathfrak X$, written $\mathfrak X\models\varphi$, if $\val\varphi_\mathfrak X =|\mathfrak X|$. If $\mathfrak X$ is a dynamical system, we write $\mathfrak X\models\varphi$ if $(\mathfrak X,\val\cdot)\models \varphi$ for every valuation $\val\cdot$ on $\mathfrak X$. If $\Omega$ is a class of dynamical systems or models, we say that $\varphi\in\lanfull$ is {\em valid on $\Omega$} if, for every $\mathfrak X\in \Omega$, $\mathfrak X\models\varphi$. If $\varphi$ is not valid on $\Omega$, it is {\em falsifiable on $\Omega$.} 

For a temporal language $\mathcal L_M$ and a class of dynamical systems $\Omega$ we define the logic ${\sf ITL}^{\Omega}_M$ to be the set of formulas of $\mathcal L_M$ that are valid over $\Omega$.
\end{definition}

As before we write ${\sf ITL}^\Omega_\nx$ instead of ${\sf ITL}^\Omega_\varnothing$. Some classes of interest are the class $\sf c$ of all dynamical systems, the class $\sf o$ of all dynamical systems with a (continuous and) open map, the class $\sf e$ of all dynamical systems based on a poset (which we call {\em expanding posets}), and the class ${\sf p} = {\sf e} \cap {\sf o}$ of {\em persistent posets.} If $\mathfrak X$ is a topological space, ${\sf ITL}^\mathfrak X_M$ denotes the set of $\mathcal L_M$-formulas valid on the class of dynamical systems of the form $(\mathfrak X, f)$.

\begin{exam}\label{examRInco}
Consider the formula $\varphi = ( \neg \nx p \wedge \nx \neg \neg p ) \rightarrow ( \nx q \vee \neg \nx q )$. Let us see that $\varphi$ is valid on $\mathbb R$ but not over all dynamical systems. Suppose that $(\mathbb R,f,\val\cdot)$ is a model based on $\mathbb R$ and that $x\in \val {\neg \nx p \wedge \nx \neg \neg p}$.
From $x\in \val { \nx \neg \neg p}$ and the semantics of double negation (discussed in \cite{FernandezITLc}) we see that there is a neighbourhood $V$ of $f(x)$ such that $V\subseteq \overline{\val p}$. It follows from the intermediate value theorem that if $U$ is a neighbourhood of $x$ and $f(U)$ is not a singleton, $f(U) \cap V$ contains an open set and hence $f(U) \cap \val p \not = \varnothing$. Meanwhile, from $x\in \val{\neg \nx p}$ we see that $x$ has a neighbourhood $U_\ast$ such that $f(U_\ast) \cap \val p = \varnothing$, hence for such a $U_\ast$ we must have that $f(U_\ast)$ is the singleton $ \{ f(x) \} $.
But then either $f(x) \in \val q$ and $x \in \val{\nx q }$, or else $f(x) \not \in \val q$, which means that $U_\ast \cap  \val{\nx q } = \varnothing$ and thus $U_\ast$ witnesses that $x\in \val{\neg \nx q}$.
In either case, $x\in \val{\nx q \vee \neg \nx q}$, as required.

On the other hand, consider the plane $\mathbb R^2$ with the projection function $\pi(x,y) = x$, and let $\val p $ be the complement of the $x$ axis and $\val q $ the complement of the $y$ axis. It is not hard to check that $0 \not \in \val \varphi$.
\end{exam}

Note that the formula $\varphi$ tells us that ${\sf ITL}^\mathbb R_\nx$ does not enjoy Craig interpolation.
The use of $\mathbb R^2$ in this example is not accidental: as we will see, any $\landi$-formula that is valid on $\mathbb R^2$ is valid over the class of all topological spaces. Note that this is no longer the case for $\varphi \in \landif$ \cite{FernandezITLc}.

\subsection{Axiomatic systems}

Our deductive calculi are obtained from propositional intuitionistic logic \cite{MintsInt} by adding standard axioms and inference rules of $\sf LTL$ \cite{temporal}, although some modifications are needed to present them in terms of $\diam$ instead of $\ubox$.
For our purposes, a {\em logic} is a set of axioms and rules defining a subset of some temporal language $\mathcal L$. We say that $\Lambda'$ {\em extends} $\Lambda$ if the language of $\Lambda'$ contains that of $\Lambda$ and $\Lambda'$ is closed under all substitution instances of the axioms and rules defining $\Lambda$.

Let us first give two axiomatizations for $\mathcal L_\nx$.
The logic $\logpersax_\nx$ (for {\em Fischer Servi}) is the least set of $\cl L_\nx$-formulas closed under the the axioms of Intuitionistic Propositional Logic~\cite{MintsInt} plus the following axioms and inference rules:

\begin{multicols}{2}
\begin{enumerate}[label=({N}\arabic*),leftmargin=*]
\item\label{ax02Bot}\mbox{$\neg \tnext \bot$}
\item\label{ax03NexWedge}\mbox{$ \tnext \varphi \wedge\tnext \psi  \rightarrow \tnext \left( \varphi \wedge \psi \right)$}
\item\label{ax04NexVee}\mbox{$ \tnext \left( \varphi \vee \psi \right) \rightarrow  \tnext \varphi \vee\tnext \psi $}
\item\label{ax05KNext}\mbox{$\tnext\left( \varphi \rightarrow \psi \right) \rightarrow \left(\tnext\varphi \rightarrow \tnext\psi\right)$}
\item\label{axFSNext}\mbox{$ \left(\tnext\varphi \rightarrow \tnext\psi\right) \rightarrow \tnext\left( \varphi \rightarrow \psi \right)$}
\end{enumerate}
\begin{enumerate}[label=({{NR}\arabic*}),leftmargin=*]
	\item\label{ax13MP} $\displaystyle\frac {\varphi \ \ \ \varphi\to \psi }{\psi}$
	\item\label{ax14NecCirc} $\displaystyle\frac{\varphi}{\tnext\varphi}$
\end{enumerate}
\end{multicols}

All of the above axioms for $\nx$ are standard for a functional modality.
We also define the logic ${\sf ITL}^0_\nx$ by omitting axiom \ref{axFSNext}, which is not valid over the class of dynamical systems, although it {\em is} valid over the class of open systems \cite{BoudouLICS}.
In contrast, we can derive the converses of the other axioms. Below, for a set of formulas $\Gamma$ we define $\tnext \Gamma = \{\tnext\varphi : \varphi \in \Gamma\}$, and empty conjunctions and disjunctions are defined by $\bigwedge\varnothing =\top$ and $\bigvee \varnothing = \bot$.

\begin{lemma}\label{lemmReverseNext}
Let $\mathcal L$ be a temporal language and $\Lambda$ be a logic over $\mathcal L$ extending ${\sf ITL}^0_\nx$. Let $\Gamma\subseteq \mathcal L$ be finite. Then, the following are derivable in $\Lambda$:
\begin{enumerate*}[label=(\arabic*)]
	\item $\tnext \bigwedge \Gamma \leftrightarrow \bigwedge \tnext \Gamma$ \hspace{20pt}
	\item $\tnext \bigvee \Gamma \leftrightarrow \bigvee \tnext \Gamma$.
\end{enumerate*}
\end{lemma}

\proof
One direction is obtained from repeated use of axioms \ref{ax03NexWedge} or \ref{ax04NexVee} and the other is proven using \ref{ax14NecCirc} and \ref{ax05KNext}; note that the second claim requires \ref{ax02Bot} to treat the case when $\Gamma = \varnothing$. Details are left to the reader.
\endproof

Next we define the logic ${\logbasic}$ by extending ${\sf ITL}^0_\nx$ with the following axioms and rules.

\medskip

\begin{center}
\begin{tabular}{ccc}
\mylabel{ax10DiamFix}{(E1)}  $\varphi \vee \tnext \diam \varphi \to \diam \varphi$\hspace{10pt} & \mylabel{ax11:dist}{(ER1)} $\displaystyle\frac{ \varphi \rightarrow  \psi } { \diam \varphi \rightarrow \diam \psi }$  & \hspace{10pt} \mylabel{ax12:ind:2}{(ER2)}  $\displaystyle\frac{ \tnext \varphi \to \varphi} { \diam \varphi \rightarrow \varphi } $
\end{tabular}
\end{center}

\medskip
%

Axiom \ref{ax10DiamFix} is the dual of $\ubox \varphi \rightarrow \varphi \wedge \tnext \ubox \varphi$.
The rule \ref{ax11:dist} replaces the dual K-axiom $\ubox(\varphi \to \psi)\to(\diam \varphi \to \diam \psi)$ common in intuitionistic modal logic, while \ref{ax12:ind:2} is dual to the induction rule $ \frac{ \varphi \to \tnext  \varphi} { \varphi \rightarrow \ubox \varphi } $.
Of course we could also consider a logic $\logpersax_\diam$ which includes axiom \ref{axFSNext}, but we do not have any completeness results for this logic.

\begin{lemma}\label{lemmReverseDiam}
Let $\mathcal L$ be a temporal language and $\Lambda$ be a logic over $\mathcal L$ extending ${\sf ITL}^0_\ps$. Then, for any $\varphi \in \mathcal L$,
$\Lambda \vdash \diam \varphi \to \varphi \vee \tnext \diam \varphi.$
\end{lemma}

\proof
Reasoning within ${\logbasic}$, note that $\varphi \to \diam \varphi$ holds by \ref{ax10DiamFix} and propositional reasoning, hence $\tnext\varphi \to \tnext \diam \varphi$ by \ref{ax14NecCirc}, \ref{ax05KNext} and \ref{ax13MP}. 
In a similar way, $\tnext \diam \varphi \to \diam \varphi$ holds by \ref{ax10DiamFix} and propositional reasoning, so $\tnext \tnext\diam \varphi \to \tnext \diam \varphi$ does by \ref{ax14NecCirc}, \ref{ax05KNext} and \ref{ax13MP}. Hence, $\tnext \varphi \vee \tnext \tnext \diam \varphi \to \tnext \diam \varphi$ holds. Using \ref{ax04NexVee} and some propositional reasoning we obtain $\tnext(\varphi \vee \tnext \diam \varphi) \to \varphi \vee \tnext \diam \varphi$. But then, by \ref{ax12:ind:2}, $\diam(\varphi \vee \tnext \diam \varphi) \to \varphi \vee \tnext \diam \varphi$; since $\diam\varphi \to \diam (\varphi \vee \tnext \diam \varphi)$ can be proven using \ref{ax11:dist}, we obtain $\diam\varphi \to \varphi \vee \tnext \diam \varphi$, as needed.
\endproof
\noindent Finally, we define the logics ${\sf ITL}^ 0_{\forall}$ and ${\sf ITL}^ 0_{\ps\forall}$ by extending ${\sf ITL}^ 0 $ and ${\sf ITL}^ 0_\ps$, respetively, with the following axioms and rule.

\begin{multicols}{2}
	\begin{enumerate}[label=(UA\arabic*),leftmargin=*]
\item\label{axUnivEM} $\forall \varphi \vee \neg \forall \varphi$
\item\label{axUnivK} $\forall (\varphi \to \psi) \to (\forall \varphi \to \forall \psi)$
\item\label{axUnivVee} $\forall (\varphi \vee \forall \psi) \to \forall \varphi \vee \forall \psi$
\item\label{axUnivT} $\forall \varphi \rightarrow \varphi$
\item\label{axUniv4} $\forall \varphi \rightarrow \forall\forall \varphi $
\item\label{axUnivNex} $\forall \varphi \leftrightarrow \nx \forall \varphi$
\end{enumerate}
\begin{enumerate}[label=({UR\arabic*}),leftmargin=*]
\item\label{rulUnivNec} $\displaystyle\frac \varphi {\forall \varphi}$
\end{enumerate}

\end{multicols}

The reader may observe that these axioms are designed to make the universal modality behave classically; indeed this is not surprising, as the only truth values that $\forall \varphi$ can take are the whole space or the empty set.
With the exception of \ref{axFSNext}, we will assume that all of the above rules and axioms are available when relevant.

\begin{definition}
An {\em admissible intuitionistic temporal logic} is any logic $\Lambda$ over a temporal language $\cl L_M$ such that $\Lambda$ extends ${\sf ITL}^0_M$.
\end{definition}

As usual, a logic $\Lambda$ is {\em sound} for a class of structures $\Omega$ if, whenever $\Lambda \vdash \varphi$, it follows that $\Omega \models \varphi$.
The following is essentially proven in \cite{BoudouLICS}:
   
\begin{theorem}\label{ThmSoundZero}
  ${\sf ITL}^0_{\ps\forall}$ is sound for the class of dynamical systems.
  \end{theorem}

Note however that a few of the axioms and rules have been modified to fall within $\cl L_\diam$, but these modifications are innocuous and their correctness may be readily checked by the reader.
We will see that every admissible intuitionistic temporal logic is also complete for the class of dynamic topological systems.

\section{Labelled structures}\label{SecNDQ}

The central ingredient of our completeness proof is given by non-deterministic {\em quasimodels,} introduced by Fern\'andez-Duque \cite{FernandezNonDeterministic} in the context of dynamic topological logic and later adapted to intuitionistic temporal logic \cite{FernandezITLc}.

\subsection{Two-sided types}

Our presentation will differ slightly from that of \cite{FernandezITLc}, since it will be convenient for us to use two-sided types, defined as follows.

\begin{definition}\label{def:type}
Let $\Sigma \subseteq \landif$ be closed under subformulas and $\Phi^+,\Phi ^ - \subseteq \Sigma$ be finite sets of formulas. We say that the pair $\Phi = (\Phi ^+ ; \Phi ^-) $ is a {\em two-sided $\Sigma$-type} if:
\begin{multicols}{2}
  \begin{enumerate}
	\item $\Phi^- \cap \Phi ^ +  = \varnothing$,
	\label{cond:type:intersection}
	
	\item $\bot\not\in \Phi ^ +$,
	\label{cond:type:bot}
	
	\item if $\varphi\wedge\psi\in \Phi ^ +$, then $\varphi,\psi\in \Phi^+$,
	\label{cond:type:posconj}
	
	\item if $\varphi\wedge\psi\in \Phi ^ -$, then	
	$\varphi \in \Phi ^-$ or $\psi\in \Phi^-$,
	\label{cond:type:negconj}
	
	\item if $\varphi\vee\psi\in \Phi ^ +$, then	$\varphi \in \Phi^+$ or $\psi\in \Phi^+$,
	
	\columnbreak
	
	\label{cond:type:posdisj}
	
	\item if $\varphi\vee\psi\in \Phi ^ -$, then  $\varphi , \psi\in \Phi^-$,
	\label{cond:type:negdisj}
	
		\item if $\varphi\to\psi\in \Phi^+$, then $\varphi \in \Phi^-$ or $\psi \in\Phi^+$,
	\label{cond:type:implication}

		\item if $\varphi\to\psi\in \Phi^-$, then $\psi \in\Phi^-$, and
	\label{cond:type:implication:neg}

	\item\label{cond:type:diam} if $\diam \varphi \in \Phi^-$ then $\varphi \in \Phi^-$.
		
\end{enumerate}
\end{multicols}
 If moreover $\Sigma = \Phi^- \cup \Phi^+$, we may say that $\Phi$ is {\em saturated.}
  The set of finite two-sided $\Sigma$-types will be denoted $\type{\Sigma}$.  
\end{definition}
  Whenever $\Xi$ is an expression denoting a two-sided type,
  we write $\Xi^+$ and $\Xi^-$ to denote its components.
We will consider three partial orders on $\type{\Sigma}$. We will write\david{Lo volv\'i a separar porque se me hac\'i dif\'icil de leer.}
\begin{enumerate}[label=\arabic*)]
  \item $\Phi \peqT \Psi$ if $\Psi^-\subseteq \Phi^-$ and $\Phi^+ \subseteq \Psi^+$,
  \item $\Phi \subseteq_T \Psi$ if $\Phi^- \subseteq \Psi^-$ and $\Phi^+\subseteq \Psi^+$,  
  and
  \item $\Phi \sqsubT \Psi$ if $\Phi^- = \Psi^-$ and $\Phi^+\subseteq \Psi^+$.
\end{enumerate}

\begin{remark}\label{RemarkTypes}
Fern\'andez-Duque \cite{FernandezITLc} uses one-sided $\Sigma$-types, but it is readily checked that a one-sided type $\Phi$ as defined there can be regarded as a saturated two-sided type $\Psi$ by setting $\Psi^+=\Phi$ and $\Psi^- = \Sigma \setminus \Phi$.
Henceforth we will write {\em type} instead of {\em two-sided type} and explicitly write {\em one-sided type} when discussing \cite{FernandezITLc}.
\end{remark}

Many times we want $\Sigma$ to be finite, and to indicate this, given $\Delta\subseteq \lanfull$ we write $\Sigma\Subset \Delta$ if $\Sigma\subseteq \Delta$ is finite and closed under subformulas.
Note that $\type\Sigma$ is partially ordered by $\subseteq$, and we will endow it with the up-set topology $\mathcal U_\subseteq$. For $\Phi\in\type\Sigma$, say that a formula $\varphi\imp \psi\in\Sigma$ is a {\em defect} of $\Phi$ if $\varphi \imp \psi \in \Phi^-$ but $\varphi \not \in \Phi^+$. The set of defects of $\Phi$ will be denoted $\defect\Phi$.

\begin{definition}\label{frame}
Let $\Sigma \subseteq \landif$ be closed under subformulas.
We say that a {\em $\Sigma$-labelled space} is a triple $\mathfrak W= ( |\mathfrak W|,\mathcal T_\mathfrak W,\ell_\mathfrak W )$, where $( |\mathfrak W|,\mathcal T_\mathfrak W )$ is a topological space and $\ell_\mathfrak W\colon | \mathfrak W | \to \type \Sigma$ a continuous function such that for all $w\in |\mathfrak W|$, whenever $\varphi\imp \psi\in \defect{ \ell_\mathfrak W(w)}$ and $U$ is any neighborhood of $w$, there is $v\in U$ such that $\varphi\in \ell_\mathfrak W(v)$ and $\psi\not\in \ell_\mathfrak W(v)$. Such a $v$ {\em revokes} $\varphi\imp \psi$.

	The $\mathcal L_M$-labelled space $\mathfrak W$ {\em satisfies} $\varphi\in\mathcal L$ if $\varphi\in \ell^+_\mathfrak W(w)$ for some $w\in |\mathfrak W|$, and {\em falsifies} $\varphi\in\mathcal L$ if $\varphi\in \ell^-_\mathfrak W(w)$ for some $w\in |\mathfrak W|$. We say that $\ell_\mathfrak W$ is {\em honest} if, for every $w\in |\mathfrak W|$ and every $\forall\varphi\in \Sigma$, we have that $\forall\varphi\in \ell^+_\mathfrak W(w)$ implies that $\varphi\in \ell^+_\mathfrak W(v)$ for every $v\in |\mathfrak W|$, and $\forall\varphi\in \ell^-_\mathfrak W(w)$ implies that $\varphi\in \ell^-_\mathfrak W(v)$ for some $v\in |\mathfrak W|$.
We say that $\fr W$ and $\ell_\mathfrak W$ are {\em saturated} if $\ell_\mathfrak W (w)$ is saturated for every $w\in |\mathfrak W|$.
\end{definition}

If $\mathfrak W$ is a labelled space, elements of $|\mathfrak W|$ will sometimes be called {\em worlds.} As usual, we may write $\ell$ instead of $\ell_\mathfrak W$ when this does not lead to confusion. Since we have endowed $\type\Sigma$ with the topology $\mathcal U_\subseteq$, the continuity of $\ell$ means that for every $w\in|\mathfrak W|$, there is a neighborhood $U$ of $w$ such that, whenever $v \in U$, $\ell_\mathfrak W(w)\subseteq \ell_\mathfrak W(v)$.

Note that not every subset $U$ of $|\mathfrak W|$ gives rise to a substructure that is also a labelled space; however, this is the case when $U$ is open. The following is not hard to check.

\begin{lemma}\label{LemmOpenSubst}
If $\Sigma\subseteq \landif$ is closed under subformulas, $\mathfrak W$ is a $\Sigma$-labelled space, and $U\subseteq |\mathfrak W|$ is open, then $\mathfrak W\upharpoonright U$ is a $\Sigma$-labelled space.
\end{lemma}

For our purposes, a {\em continuous relation} on a topological space is a relation under which the preimage of any open set is open; note that this is {\em not} the standard definition of a contiuous relation.
In the context of an Alexandroff space with the up-set topology, a continuous relation $S$ is one that satisfies the {\em forward confluence} property: if $w' \seq w \mathrel S v$, then there is $v'$ such that $w'\mathrel S v' \seq v$.
Similarly, an open relation $S$ is one such that if $ w \mathrel S v \peq v'$, then there is $w'$ such that $w \peq w' \mathrel S v' $.

\begin{definition}\label{compatible}
Let $\Sigma \subseteq \landif$ be closed under subformulas and $\Phi,\Psi\in\type \Sigma$. The ordered pair $(\Phi,\Psi)$ is {\em sensible} if

\begin{multicols}{2}
\begin{enumerate}[label=\arabic*)]
	\item\label{ItCompOne} $\tnext\varphi\in \Phi^+$ implies $ \varphi\in \Psi^+$,
	\item\label{ItCompTwo} $\tnext\varphi\in \Phi^-$ implies $ \varphi\in \Psi^-$,
	\item\label{ItCompThree}$\diam\varphi\in \Phi^+$ implies

$\varphi\in\Phi^+$ or $\diam\varphi\in \Psi^+$,
	\columnbreak
	\item\label{ItCompFour} $\diam\varphi\in \Phi^-$, implies $\diam \varphi \in \Psi^-$,
	\item\label{ItCompFive} $\forall\varphi\in \Phi^+$ iff $\forall\varphi\in \Psi^+,$ and
	\item\label{ItCompSix} $\forall\varphi\in \Phi^-$ iff $\forall \varphi \in \Psi^-$.
\end{enumerate}
\end{multicols}

Likewise, a pair $(w,v)$ of worlds in a labelled space $\mathfrak W$ is sensible if $(\ell (w),\ell (v))$ is sensible.

A continuous relation
$S\subseteq |\mathfrak W|\times |\mathfrak W|$
is {\em sensible} if every pair in $S$ is sensible.
Further, $S$ is {\em $\omega$-sensible} if it is serial and, whenever $\ps\varphi\in \ell(w)$, there are $n\geq 0$ and $v$ such that $w \mathrel S^n v$ and $\varphi\in \ell(v)$.

A {\em $\Sigma$-labelled system} is a $\Sigma$-labelled space $\mathfrak W$ equipped with a sensible relation $S_\mathfrak W\subseteq |\mathfrak W|\times|\mathfrak W|$; if moreover $\ell_\mathfrak W$ is honest and $S_\mathfrak W$ is $\omega$-sensible, we say that $\mathfrak W$ is a {\em well $\Sigma$-labelled system.}
\end{definition}

Given $\Sigma \subseteq \landif$ closed under subformulas, any dynamic topological model can be regarded as a well $\Sigma$-labelled system. If $\mathfrak{X}$ is a model, we can assign a saturated $\Sigma$-type $\ell_\mathfrak X(x)$ to $x$ given by $\ell_\mathfrak X(x)=\cbra\psi\in \Sigma :x\in \val\psi_\mathfrak X \cket.$ We also set $S_\mathfrak X=f_\mathfrak X$; it is obvious that $\ell_\mathfrak X$ is honest and $S_\mathfrak X$ is $\omega$-sensible.
Henceforth we will tacitly identify $\mathfrak X$ by its associated well $\landif$-labelled system.
However, not all labelled systems we are interested in arise from models: another useful class of labelled systems is given by quasimodels.

\begin{definition}\label{ndqm}
Given $\Sigma \subseteq \landif$ closed under subformulas, a {\em weak $\Sigma$-quasimodel} is a $\Sigma$-labelled system $\mathfrak Q$ such that $\mathcal T_\mathfrak Q$ is equal to the up-set topology for a partial order which we denote $\peq_\mathfrak Q$. If moreover $\mathfrak Q$ is a well $\Sigma$-labelled system, then we say that $\mathfrak{Q}$ is a {\em $\Sigma$-quasimodel.}
\end{definition}

Note that quasimodels are very close to models, except that the relation $S$ may be non-deterministic. Indeed, deterministic quasimodels are essentially models.
  The following can be checked by a standard structural induction on $\varphi$.

\begin{lemma}\label{LemTruth}
  Let $\Sigma \subseteq \cl L_{\ps \forall}$ be closed under subformulas and $\fr Q$ be an honest, deterministic $\Sigma$-quasimodel.
  
  Define a valuation $\val\cdot_\fr Q$ on $ {\fr Q}$ by setting
$\val p_\fr Q =\{w \in W : p \in \ell(w)^+\}$
and extending to all of $\mathcal L$ recursively.
Then, for all formulas $\varphi \in \mathcal L_\diam$ and for all $w \in { W}$,
	\begin{enumerate*}[label=\arabic*)]
		\item if $\varphi \in \ell(w)^+$ then  $w \in \val{\varphi}_\fr Q$, and
		\item if $\varphi \in \ell(w)^-$ then  $w \not \in \val{\varphi} _ \fr Q$.
	\end{enumerate*}
\end{lemma}

\fullproof{
\proof
We proceed by structural induction on $\varphi$. We must consider the following cases.
\medskip

\noindent ($\varphi = p$ is an atom) \ Note that by definition of $\val{p}^{\ell}$, if $p \in {\ell}^+(w)$ then $w \in \val{p}^{\ell}$ and if $p \in {\ell}^-(w)$ then $p \not \in {\ell}^+(w)$ so $w \not \in \val{p}^{\ell}$.
\medskip

\noindent ($\varphi = \psi \wedge \theta$) \ Assume that $\psi \wedge \theta \in {\ell}^+(w)$. By Definition~\ref{def:type} it follows that $\psi \in {\ell}^+(w)$ and $\theta \in {\ell}^+(w)$. By induction hypothesis,  $w \in \val{\psi}^{\ell}$ and $w \in \val{\theta}^{\ell}$. Therefore $w \in \val{\psi \wedge \theta}^{\ell}$.

If $\psi \wedge \theta \in {\ell}^-(w)$, by definition~\ref{def:type} it follows that either $\psi \in {\ell}^-(w)$ or $\theta \in {\ell}^-(w)$. By induction hypotheses we conclude that $w \not \in \val{\psi}^{\ell}$ or $w \not \in \val{\theta}^{\ell}$. Therefore $w \not \in \val{\psi \wedge \theta}^{\ell}$.
\medskip

\noindent ($\varphi = \psi \vee \theta$) This case is symmetric, but using the conditions for $\vee$.
\medskip

\noindent ($\varphi = \psi \to \theta$) \
Assume first that $\psi \to \theta \in {\ell}^+(w)$.
Then for all $y$ such that $w \peq y$, by condition~\ref{cond:frame:monotony} of Definition~\ref{frame}, $\psi \to \theta \in \ell(y)^+$.
By condition~\ref{cond:type:implication} of Definition~\ref{def:type} and by induction hypothesis, $y \notin {\val \psi}_\fr Q$ or $y \in {\val \theta}_\fr Q$.
Therefore, $w \in {\val{\psi \to \theta}}_\fr Q$.

Now let us assume that $\psi \to \theta \in {\ell}^-(w)$. By Definition~\ref{frame} it follows that there exists $v \in W$ such that $w\peq v$ and $\psi \in {\ell}^+(v)$ and $\theta \in {\ell}^-(v)$. By induction hypothesis it follows that $v  \in \val{\psi}^{\ell} \setminus \val{\theta}^{\ell}$, which means that $w \not \in \val{\psi\to \theta}^{\ell}$.
\medskip

\noindent ($\varphi = \tnext \psi$) \ Assume that $\tnext \psi \in {\ell}^+(w)$. Since $S$ is sensible, $\psi \in {\ell}^+(S(w))$. By induction hypothesis $S(w) \in \val{\psi}^{\ell}$. Therefore $w \in \val{\tnext \psi}^{\ell}$. The case where $\tnext \psi \in {\ell}^-(w)$ is analogous.
\medskip

\noindent ($\varphi = \diam \psi$) If $\diam \psi \in {\ell}^+(w)$, by the fact that $S$ is $\omega$-sensible there exists $v\in W$ such that $w \mathrel S^n v$ and $\psi \in {\ell}^+(v)$; since $S$ is deterministic, we must forcibly have $v=S^n(w)$. By induction hypothesis we conclude that $v \in \val{\psi}^{\ell}$ and by the satisfaction relation it follows that $w \in \val{\diam \psi}^{\ell}$.

In case that $\diam \psi \in {\ell}^-(w)$, observe that for all $n$, if $\diam \psi \in {\ell}^-(S^n(w))$ then $\diam \psi \in {\ell}^-(S^{n+1}(w))$; thus by induction, $\diam \psi \in {\ell}^-(S^n(w))$ for all $n<\omega$. In virtue of Definition \ref{def:type}.\ref{cond:type:diam}, $\psi \in {\ell}^-(S^n(w))$ for all $n<\omega$, hence by the induction hypothesis $S^n(w) \not \in \val \psi$, from which it follows that $w \not \in \val{\diam \psi}$.
\endproof
}

In the non-deterministic case quasimodels are not models as they stand, but in \cite{FernandezITLc}, it is shown that dynamical systems can be extracted from them.

\begin{theorem}[Fern\'andez-Duque \cite{FernandezITLc}]\label{TheoITLc}
A formula $\varphi \in \landif$ is satisfiable (falsifiable) over the class of dynamic topological systems if and only if it is satisfiable (falsifiable) over the class of saturated, finite, ${\rm sub}(\varphi)$-quasimodels.
\end{theorem}

Note that \cite{FernandezITLc} uses quasimodels with one-sided types, but in view of Remark \ref{RemarkTypes}, the theorem can easily be modified to obtain quasimodels with two-sided types.
Two-sided types will be more convenient for us, especially in Section \ref{SecCons}.
Below, recall that for a structure $\fr A$ and $U\subseteq |\fr A|$, $\fr A \upharpoonright U$ is the substructure of $\fr A$ obtained by restricting all functions and relations of $\fr A$ to $U$.

\begin{lemma}\label{LemmIsQuasi}
Let $\fr Q$ be a (weak) quasimodel and $U \subseteq |\fr Q|$ be open. If either
\begin{enumerate*}[label=\arabic*)]
\item $S_{\fr Q} \upharpoonright U$ is serial and $\omega$-sensible, or
\item $U$ is $S _{\fr Q} $-invariant (i.e., $S _{\fr Q} (U) \subseteq U$),
\end{enumerate*}
then $\fr Q \upharpoonright U$ is a (weak) quasimodel.
\end{lemma}

\fullproof{
\proof We must show that ${\cl Q}$ satisfies all properties of Definition~\ref{def:quasimodel}. First we check that
$(U , \mathord{\peq}_{{\cl Q}\upharpoonright U}, \ell_{{\cl Q}\upharpoonright U})$
is a labelled frame.
The relation $\mathord{\peq}_{{\cl Q}\upharpoonright U}$ is a partial order, since restrictions of partial orders are partial orders.
Similarly, if $x \peq_{\cl Q\upharpoonright U} y$ it follows that
$x \peq_{{\cl Q}} y$, so that from the definition of
$\ell_{{\cl Q}\upharpoonright U}$ it is easy to deduce that
$\ell_{{\cl Q}\upharpoonright U}(x) \peqT \ell_{{\cl Q}\upharpoonright U}(y)$.

To check that condition~\ref{cond:frame:imp} holds, let us take $x \in U$ and a formula $\varphi \to \psi \in \ell_{{\cl Q}\upharpoonright U}^-(x)$. By definition, $\varphi \to \psi \in \ell_{{\cl Q}}^-(x)$ so there exists $y \in |{\cl Q}|$ such that $x \peq_{\cl Q} y$, $\varphi \in \ell_{\cl Q}^+(y)$ and $\psi \in \ell_{\cl Q}^-(y)$. Since $U$ is upward closed then $y\in U$ and, by definition, $x  \peq_{{\cl Q}\upharpoonright U} y$, $\varphi \in \ell_{{\cl Q}\upharpoonright U}^+(y)$ and $\psi \in \ell_{{\cl Q}\upharpoonright U}^-(y)$, as needed.
\medskip

Now we check that the relation $S_{\cl Q\upharpoonright U}$ satisfies \eqref{itSerial}-\eqref{itOmega}.
Note that $S_{\cl Q\upharpoonright U}$ is serial and $\omega$-sensible by assumption and it is clearly sensible as $S_\cl Q$ was already sensible, so it remains to see that $S_{{\cl Q}\upharpoonright U}$ is forward-confluent. Take $x,y, z \in U$ such that $ x  \peq_{{\cl Q}\upharpoonright U} y$ and $x \mathrel S_{{\cl Q}\upharpoonright U} z$. By definition $x  \peq _{\cl Q} y$ and $x \mathrel S_{\cl Q} z$. Since $S_{\cl Q}$ is confluent, there exists $t \in |{\cl Q}|$ such that $ z  \peq _{\cl Q} t$ and $y \mathrel S_{\cl Q} t$. Since $U$ is upward closed $t \in U$ and, by definition, $y \mathrel S_{{\cl Q}\upharpoonright U} t$ and $z \mathord{\peq}_{{\cl Q}\upharpoonright U} t$.
\endproof
}

\shortproof{
\proof
By Lemma \ref{LemmOpenSubst} we know that $\fr Q \upharpoonright U$ is a labelled frame, while $S_{\fr Q \upharpoonright U}$ is clearly sensible.
Since $U$ is open and $S_\fr Q$ is continuous, $S_{\fr Q \upharpoonright U}$ is continuous as well.
Thus it remains to show that $S_{\fr Q\upharpoonright U}$ is serial and $\omega$-sensible, which in the first case holds by assumption and in the second follows easily from $S_\fr Q$ already having these properties.
\endproof
}

As usual, if $\varphi$ is not derivable, we wish to produce a model where $\varphi$ is falsified, but in view of Theorem \ref{TheoITLc}, it suffices to falsify $\varphi$ on a quasimodel. This is convenient, as quasimodels are much easier to construct than models.

\section{The canonical model}\label{secCanMod}

In this section we construct a standard canonical model for any logic $\Lambda$ extending ${\sf ITL}^0_\nx$.
From this we will obtain some completeness results for logics over $\mathcal L_\nx$.
However, in the presence of $\ps$, the standard canonical model is only a saturated, weak, deterministic quasimodel rather than a proper model.
Nevertheless, the canonical model will later be a useful ingredient in our completeness proofs for ${\sf ITL}^0_\ps$ and ${\sf ITL}^0_{\ps \forall}$.
Since we are working over an intuitionistic logic, the role of maximal consistent sets will be played by prime types, as defined below. 

\begin{definition}\label{def:prime}
Let $\mathcal L$ be a temporal language and $\Lambda$ a logic over $\mathcal L$.
Given two sets of formulas $\Gamma$ and $\Delta$, we say that $\Delta$ is a consequence of $\Gamma$ (with respect to $\Lambda$), denoted by $\Gamma \vdash \Delta$, if there exist finite $\Gamma'\subseteq \Gamma$ and $\Delta' \subseteq \Delta$ such that $ \Lambda \vdash \bigwedge \Gamma' \to \bigvee \Delta'$.

We say that a pair of sets $\Phi =(\Phi^+,\Phi^-)$ is {\em consistent} if $\Phi^+ \not\vdash \Phi^-$. A saturated, consistent $\mathcal L$-type is a {\em prime type.} The set of prime $\mathcal L$-types will be denoted $\ptypel {\mathcal L}$.
\end{definition}

Note that we are using the standard interpretation of $\Gamma \vdash \Delta$ in Gentzen-style calculi.
The logic $\Lambda$ will always be clear from context, which is why we do not reflect it in the notation.
When working within a turnstyle, we will follow the usual proof-theoretic conventions of writing $\Gamma,\Delta$ instead of $\Gamma \cup \Delta$ and $\varphi$ instead of $\{\varphi\}$.
Observe that there is no clash in terminology regarding the use of the word {\em type:}

\begin{lemma}\label{lemmPrimeIsType}
If $\Lambda$ is an admissible temporal logic over a language $\mathcal L$ and $\Phi$ is a prime $\mathcal L$-type then $\Phi$ is an $\cl L$-type.
\end{lemma}

\begin{proof}
Let $\Phi$ be a prime $\cl L$-type. Observe that $\Phi$ is already saturated by definition, so it remains to check that it satisfies all conditions of Definition \ref{def:type}.

Conditions \ref{cond:type:intersection} and \ref{cond:type:bot} follow from the consistency of $\Phi$.
The proofs of the other conditions are all similar to each other. For example, for \ref{cond:type:implication}, suppose that $\varphi \to \psi \in \Phi^+$ and $\varphi \not \in \Phi^-$. Since $\Phi$ is saturated, it follows that $\varphi \in \Phi^+$. But $\big (\varphi \wedge (\varphi \to \psi)\big) \to \psi$ is an intuitionistic tautology, so using the fact that $\Phi$ is consistent we see that $\psi \not \in \Psi^-$, which using the assumption that $\Phi$ is saturated gives us $\psi \in \Phi^+$.
For condition \ref{cond:type:diam} we use \ref{ax10DiamFix}: if $\diam \varphi \in \Phi^-$ and $\varphi \in \Phi^+$ we would have that $\Phi$ is inconsistent, hence $\varphi \in \Phi^-$. The rest of the conditions are left to the reader.
\end{proof}

As with maximal consistent sets, prime types satisfy a Lindenbaum property.

\begin{lemma}[Lindenbaum Lemma]\label{LemmLind}
Fix an admissible temporal logic $\Lambda$ over $\mathcal L$.	Let $\Gamma,\Delta \subseteq \mathcal L$. If $\Gamma \not\vdash \Delta$ then there exists a prime type $\Phi$ such that $\Gamma \subseteq \Phi^+$ and $\Delta \subseteq \Phi^-$.
\end{lemma}

\proof
The proof is standard, but we provide a sketch.
Let $\varphi \in \cl L $. Note that either $\Gamma ,\varphi  \not\vdash \Delta $ or $\Gamma  \not \vdash \Delta,\varphi$, for otherwise by a cut rule (which is intuitionistically derivable) we would have $\Gamma \vdash \Delta$. Thus we can add $\varphi$ to $\Gamma \cup \Delta$, and by repeating this process for each element of $\cl L_\diam$ (or using Zorn's lemma) we can find suitable $\Phi$.
\endproof

Given a set $A$, let $\mathbb I_A$ denote the identity function on $A$. Let $\Lambda$ be an admissible temporal logic over $\cl L$. The canonical model $\CMod$ for $\Lambda$ is defined as the labelled structure
\[\CMod = (|\CMod|,{\peq_\CIcon },S_\CIcon ,\ell_\CIcon ) \eqdef  (\type{\cl L},{\peq_T},S_T,{\mathbb I}_{\ptypel{\cl L}})\upharpoonright \ptype;\]
in other words, $\CMod$ is the set of prime types with the usual ordering and successor relations. Note that $\ell_{\CIcon}$ is just the identity (i.e., $\ell_\CIcon (\Phi) = \Phi$).
We will usually omit writing $\ell_\CIcon $, as it has no effect on its argument.

Next we show that $\CMod$ is a saturated, weak, deterministic quasimodel. For this, we must prove that it has all the required properties.

\begin{lemma}
	\label{lemm:normality} Let $\Lambda$ be an admissible temporal logic over $\cl L$. Then, $\M_\CIcon $ is a labelled frame.
\end{lemma}
\begin{proof}
	We know that $\peqT$ is a partial order and restrictions of partial orders are partial orders, so $\peq_\CIcon $ is a partial order. Moreover, $\ell_\CIcon $ is the identity, so $\Phi \peq_\CIcon  \Psi$ implies that $\ell_\CIcon  (\Phi) \peqT \ell_\CIcon  (\Psi)$.
	
	Now let $\Phi \in |\CMod|$ and assume that $\varphi \to \psi \in \Phi^-$. Note that $ \Phi^+,\varphi \not \vdash \psi$, for otherwise by intuitionistic reasoning we would have $\Phi^+\vdash \varphi \to \psi$, which is impossible if $\Phi$ is a prime type. By Lemma \ref{LemmLind}, there is a prime type $\Psi$ with $\Phi^+ \cup \{\varphi\} \subseteq \Psi^+$ and $\psi \in \Psi^-$. It follows that $\Phi \peq_\CIcon  \Psi$, $\varphi \in \Psi^+$ and $\psi \in \Psi^-$, as needed.
\end{proof}

\begin{lemma}
	\label{lemm:rcnext:prop} Let $\Lambda$ be an admissible temporal logic over $\cl L$.  Then, $S_{\CIcon }$ is a continuous function.
	If moreover \ref{axFSNext} is an axiom of $\Lambda$, then $S_{\CIcon }$ is also open.
\end{lemma}
\begin{proof} 
	For a set $\Gamma \subseteq \cl L_\diam$, recall that we have defined $\tnext \Gamma = \{\tnext \varphi :   \varphi \in \Gamma\}$. It will be convenient to introduce the notation $\remc \Gamma = \{\varphi : \tnext \varphi \in \Gamma\}$.
	With this, we show that $S_\CIcon$ is functional and forward-confluent.
	\medskip
	
	\noindent{\sc Functionality.} We claim that for all $\Phi, \Psi \in |\CMod|$,
	\begin{equation}\label{EqSc}
	\Phi \mathrel S_\CIcon \Psi \text{ if and only if }\Psi = (\remc \Phi^+,\remc\Phi^-). 
	\end{equation}
	We must check that $\Psi \in |\CMod|$. To see that $\Psi$ is saturated, let $\varphi \in \cl L_\diam$ be so that $\varphi \not \in \Psi^-$. It follows that $\tnext\varphi \not \in \Phi^-$, but $\Phi$ is saturated, so $\tnext\varphi \in \Phi^+$ and thus $\varphi \in \Psi^+$. Since $\varphi$ was arbitrary, $\Psi^-\cup \Psi^+ = \cl L_\diam$.
	
	Next we check that $\Psi$ is consistent. If not, let $\Gamma \subseteq \Psi^+$ and $\Delta\subseteq \Psi^-$ be finite and such that $\bigwedge \Gamma \to \bigvee \Delta$ is derivable. Using \ref{ax14NecCirc} and \ref{ax05KNext} we see that $\tnext \bigwedge \Gamma \to  \tnext \bigvee \Delta$ is derivable, which in view of Lemma \ref{lemmReverseNext} implies that $ \bigwedge \tnext \Gamma \to  \bigvee  \tnext \Delta$ is derivable as well. But $\tnext \Gamma \subseteq \Phi^+$ and $\tnext \Delta \subseteq \Phi^-$, contradicting the fact that $\Phi$ is consistent.
	
	Thus $\Psi \in |\CMod|$, and $\Phi \mathrel S_\CIcon \Psi$ holds provided that $\Phi \ST \Psi$. It is clear that clauses \ref{ItCompOne} and \ref{ItCompTwo} of Definition \ref{compatible} hold. If $\diam \varphi \in \Phi^+$ (so that $\ps \in M$) and $\varphi \not \in \Phi^+$, it follows that $\varphi \in \Phi^-$. By Lemma \ref{lemmReverseDiam} $\diam\varphi \to \varphi \vee \tnext \diam \varphi$ is derivable, so we cannot have that $\tnext \diam \varphi \in \Phi^-$ and hence $\tnext \diam \varphi \in \Phi^+$, so that $\diam \varphi \in \Psi^+$. Similarly, if $\diam \varphi \in \Phi^-$ we have that $\tnext\diam\varphi \in \Phi^-$, for otherwise we obtain a contradiction from \ref{ax10DiamFix}. Therefore, $\diam\varphi \in \Psi^-$ as well. The clauses for $\forall \varphi$ follow a similar line of reasoning using \ref{axUnivNex}.
	
	To check that $\Psi$ is unique, suppose that $\Theta \in |\CMod|$ is such that $\Phi \mathrel S_\CIcon  \Theta$. Then if $\varphi \in \Psi^+$ it follows from \eqref{EqSc} that $\tnext \varphi \in \Phi^+$ and hence $\varphi \in \Theta^+$; by the same argument, if $\varphi \in \Psi^-$ it follows that $\varphi \in \Theta^-$, and hence $\Theta = \Psi$.
	\medskip
	
	\noindent {\sc Continuity:} Now that we have shown that $S_\CIcon$ is a function, we may treat it as such. Suppose that $\Phi \peq_\CIcon \Psi$; we must check that $S_\CIcon (\Phi) \peq_\CIcon S_\CIcon (\Psi)$. Let $\varphi \in S^+_\CIcon (\Phi)$. Using \eqref{EqSc}, we have that $\tnext\varphi \in \Phi^+$, hence $\tnext\varphi \in \Psi^+ $ and thus $\varphi \in S_\CIcon (\Psi^+)$. Since $\varphi \in S_\CIcon (\Phi)$ was arbitrary we obtain $S^+_\CIcon (\Phi) \peq_\CIcon S^+_\CIcon (\Psi)$, as needed.\\
	
		\noindent {\sc Openness:} Suppose that $\Psi \seq_\CIcon S_\CIcon (\Phi)$.
		We claim that $\Theta' : = ( \tnext \Psi^+ \cup \Phi^+, \tnext \Psi^-)$ is consistent.
If not, there are finite $\Xi \subseteq \Psi^+ $, $\Gamma \subseteq \Phi^+$ and $ \Delta \subseteq \Psi^-$ such that $\vdash \bigwedge \tnext \Xi \wedge \bigwedge \Gamma \to \bigvee\tnext \Delta$.
Since $\Gamma \subseteq \Phi^+$, this yields $ \bigwedge \tnext \Xi  \to \bigvee\tnext \Delta \in \Phi^+$. Using Lemma \ref{lemmReverseNext}, \ref{axFSNext} and propositional reasoning, this gives us $\tnext \left ( \bigwedge   \Xi  \to \bigvee  \Delta \right ) \in \Phi^+$, hence $\bigwedge   \Xi  \to \bigvee  \Delta \in \Psi^+$.
But then $\bigwedge \Xi \wedge \left ( \bigwedge   \Xi  \to \bigvee  \Delta \right ) \to \bigvee \Delta$ would be an intuitionistic tautology witnessing that $\Psi$ is inconsistent, contrary to our assumption. We conclude that $\Theta' $ is consistent, hence it can be extended to a prime type $\Theta$ using the Lindenbaum lemma, and clearly $\Phi \peq_\CIcon \Theta \mathrel S_\CIcon \Psi$, as required. 
\end{proof}

\begin{proposition}\label{prop:CisW}
Let $\Lambda$ be an admissible temporal logic over $\cl L $.
Then, the canonical model for $\Lambda$ is a deterministic weak quasimodel.
\end{proposition}

\proof
We need
\begin{enumerate*}
\item $(|\CMod|,{\peq}_\CIcon ,\ell_\CIcon )$ to be a labelled frame,
\item $S_\CIcon $ to be a sensible forward-confluent function, and
\item $\ell_\CIcon $ to have $\type{\cl L }$ as its codomain.
\end{enumerate*}
The first item is Lemma \ref{lemm:normality}. That $S_\CIcon $ is a forward-confluent function is Lemma \ref{lemm:rcnext:prop}, and it is sensible since $\Phi \mathrel S_\CIcon  \Psi$ precisely when $\Phi \ST \Psi$. Finally, if $\Phi \in |\CMod|$ then $\ell_\CIcon (\Phi) = \Phi$, which is an element of $\type{\cl L }$ by Lemma \ref{lemmPrimeIsType}.
\endproof

From this we may already obtain our first completeness results.

\begin{theorem}\label{theoCircComp}
${\sf ITL}^0_{\nx }$ is complete for the class of expanding posets and $\logpersax_{\nx }$ for the class of persistent posets.
\end{theorem}

\proof
Let $\Lambda$ be either ${\sf ITL}^0_{\nx }$ or $\logpersax_{\nx }$. By the Lindenbaum lemma \ref{LemmLind}, if $\Lambda \not \vdash \varphi$ then there is $\Phi \in |\CMod|$ such that $\varphi \in \Phi^-$. By Proposition \ref{prop:CisW} $\CMod $ is a deterministic, weak quasimodel falsifying $\varphi$, and moreover it is trivially $\omega$-sensible as $\diam$ is not in our language.
By Lemma \ref{LemTruth}, $\varphi$ is not valid over the class of expanding posets, as required.
In the case that $\Lambda = \logpersax_{\nx }$, we additionally use the fact that $S_\CIcon$ is open, so that $\CMod$ is persistent.
\endproof

\section{Simulation formulas}\label{SecSim}

Simulations are relations between labelled spa\-ces, and give rise to the appropriate notion of `substructure' for modal and intuitionistic logics.
We have used them to prove that ${\sf ITL}^{\sf c}_{\ps \forall}$ has the finite quasimodel property \cite{FernandezITLc}, and they will also be useful for our completeness proof.
Below, recall that $\Phi \subT \Psi$ means that $\Phi ^ - \subseteq \Psi^-$  and $\Phi ^ + \subseteq \Psi^+$.

\begin{definition}
Let $\Sigma\subseteq \Delta \subseteq \landif$ be closed under subformulas, $\fr X$ be a $\Sigma$-labelled space and $\fr Y$ be $\Delta$-labelled. A continuous relation
${\simrel} \subseteq |{\fr X}|\times |{\fr Y}|$
is a {\em simulation} if, whenever $x\simrel y$, $\ell_{\fr X} (x) \subT  \ell_{\fr Y}(y).$
If there exists a simulation $\simrel$ such that $x \simrel y$, we write $(\fr X , x)\simu (\fr Y, y)$.

The relation $\simrel$ is a {\em dynamic simulation} between $\fr X$ and $\fr Y$ if ${
\mathrel S_\fr Y\simrel} \subseteq { \simrel \mathrel S_\fr X}
$.

\end{definition}

\begin{figure}

\begin{center}

\begin{tikzpicture}[scale=.6]

\draw[thick] (0,0) circle (.35);

\draw (.04,-.05) node {$y$};

\draw[very thick,->] (.5,0) -- (2.5,0);

\draw (1.5,-.5) node {$S_\fr Y$};

\draw[thick] (3,0) circle (.35);

\draw (3+.03,0) node {$y'$};

\draw[thick] (0,3) circle (.35);

\draw (.04,3-.04) node {$x$};

\draw[very thick,->,dashed] (.5,3) -- (2.5,3);

\draw (1.5,3.5) node {$S_\fr X$};

\draw[very thick,->] (0,2.5) -- (0,.5);

\draw (-.5,1.5) node {${\simrel}$};

\draw[thick,dashed] (3,3) circle (.35);

\draw (3+.03,3) node {$x'$};

\draw[very thick,->,dashed] (3,2.5) -- (3,.5);

\draw (3.5,1.5) node {${\simrel}$};

\end{tikzpicture}

\end{center}

\caption{If ${\simrel} \subseteq |\fr X| \times |\fr Y|$ is a dynamical simulation, this diagram can always be completed.}

\end{figure}
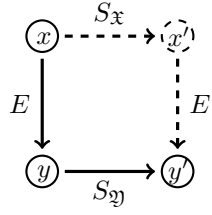

\ignore{
The following properties are readily verified:

\begin{lemma}\label{LemmPropSim}

Let ${\fr X},{\fr Y},{\fr Z}$ be labelled systems and $\simrel\subseteq |{\fr X}|\times  |{\fr Y}|$, $\xi\subseteq |{\fr Y}|\times  |{\fr Z}|$ be simulations. Then:

\begin{enumerate}

\item\label{ItPropSimComp} $\xi\simrel\subseteq |{\fr X}|\times  |{\fr Z}|$ is a simulation. Moreover, if both $\simrel$ and $\xi$ are dynamic, then so is $\xi\simrel$.
\item \label{ItPropSinSub} If $U\subseteq |{\fr X}|$ and $V\subseteq |{\fr Y}|$ are open, then $\simrel\upharpoonright U\times V$ is a simulation.
\item\label{ItPropSimUn} If $\Xi\subseteq\mathcal P(|{\fr X}|\times  |{\fr Y}|)$ is a set of simulations, then $\bigcup \Xi$ is also a simulation.

\end{enumerate}

\end{lemma}
}

Next we show that there exist formulas defining points in finite frames up to simulability, i.e.~that if $\fr W$ is a finite frame and $w\in |\fr W|$, there exists a formula $\Sim w$ such that for all labelled frames $\fr M$ and all $x\in |\fr M|$, $\fr M,x  \models x$ if and only if $(\fr W,w) \simu (\fr M,x)$.
In contrast, simulability formulas for finite $\sf S4$ models are not definable in the classical modal language \cite{FernandezSimulability}, but they {\em can} be constructed using a polyadic extension of the modal language representing the {\em tangled closure} of a family of sets \cite{FernandezTangle,GoldblattH17,GoldblattTangleSL} and expressively equivalent to the $\mu$-calculus over $\sf S4$ frames \cite{do}.

Fern\'andez-Duque \cite{dtlaxiom} uses simulation formulas to axiomatize the resulting polyadic extension of $\sf DTL$; in contrast, the natural axiomatization suggested by Kremer and Mints \cite{kmints} of dynamic topological logic is incomplete \cite{FernandezNonFin}.
In the intuitionistic setting the situation is simplified somewhat, as finite frames \cite{JonghY09} (and hence models) are already definable up to simulation in the intuitionistic language. This may be surprising, as the intuitionistic language is less expressive than the modal language; however, intuitionistic models are posets rather than arbitrary preorders, and this allows us to define simulability formulas by recursion on $\prec$.

\begin{definition}
Fix $\Sigma\Subset \landif$ and let $\fr W$ be a finite $\Sigma$-labelled frame. Given $w\in |\fr W|$, we define a formula $\Sim w$ by backwards induction on ${\peq} = {\peq_\fr W}$ by
\[\Sim w = \bigwedge \ell^+(w) \rightarrow \bigvee \ell^-(w) \vee \bigvee_{v\succ w} \Sim v .\]
\end{definition}

\begin{remark}
Observe that if $\cl L$ is a temporal language and $\fr W$ a finite $\Sigma$-labeled frame with $\Sigma \Subset \cl L$, then $\Sim w \in \cl L$ for all $w\in |\fr W|$.
\end{remark}

\begin{proposition}\label{propSimForm}
Given $\Sigma \Subset \Delta \subseteq \landif$, a finite $\Sigma$-labelled frame $\fr W$, a $\Delta$-labelled frame $\fr X$ and $w\in |\fr W|$, $x \in |\fr X |$:
\begin{enumerate}

\item\label{simulability:c1} if $\Sim w \in \ell_\fr X^-(x)$ then there is $y\seq x$ such that $(\fr W,w) \simu (\fr X, y)$, and

\item\label{simulability:c2} if there is $y\seq x$ such that $ (\fr W,w) \simu (\fr X, y) $ then $\Sim w \not \in \ell_\fr X^+(x)$.

\end{enumerate}
\end{proposition}

\proof Each claim is proved by backward induction on $\peq$.
\medskip

\noindent \eqref{simulability:c1} Let us first consider the base case, when there is no $v \succ w$. Assume that $\Sim w \in \ell^-(x)$. From the definition of labelled frame $\bigwedge \ell_\fr W^+(w) \in \ell_\fr X^+(y)$ and  $\bigvee \ell_\fr W^-(w) \in \ell_\fr X^-(y)$ for some $y \seq x$. From the definition of type it follows that $\ell_\fr W^+(w) \subseteq \ell_\fr X^+(y)$ and $\ell_\fr W^-(w) \subseteq \ell_\fr X^-(y)$, so that $\ell_\fr W(w) \subT \ell_\fr X(y)$. It follows that ${\simrel} \eqdef \lbrace (w, y)\rbrace$ is a simulation, so $(\fr W, w) \simu  ({\fr X}, y )$.

For the inductive step, let us assume that the lemma holds for all $v \succ w$.
Assume that $\Sim w \in \ell_\fr X^-(x)$. From the definition of labelled frame, it follows that $\bigwedge \ell_\fr W^+(w) \in \ell_\fr X^+(y)$,  $\bigvee \ell_\fr W^-(w) \in \ell_\fr X^-(y)$ and $\bigvee _{v \prec w} \Sim v \in \ell_\fr X^-(y)$ for some $y \seq x$.
Following similar reasoning as in the base case we can conclude that $\ell_\fr W(w) \subseteq  \ell_\fr X(y)$, and moreover, that $\Sim v \in \ell_\fr X^-(y)$ for all $v \succ w$. By induction hypothesis we conclude that for all $v \succ w$, there exists a simulation $\simrel_{v}$ such that $v \simrel_{v} z_{v}$ for some $z_{v} \seq y$. Let ${\simrel} \eqdef \lbrace (w, y) \rbrace \cup \bigcup \limits_{v \succ w} \simrel_{v}$ . The reader may check that $\simrel$ is a simulation and that $w \simrel y \seq x$, so that $(\fr W,w) \simu ({\fr X}, y)$, as needed.
\medskip

\noindent \eqref{simulability:c2}
For the base case, assume that $(\fr W,w) \simu (\fr X, y)$ for some $y \seq x$, so there exists a simulation $\simrel$ such that $w \simrel y  $. It follows that $\ell_\fr W^+(w) \subseteq \ell_\fr X^+(y)$ and $\ell_\fr W^-(w) \subseteq \ell_\fr X^-(y)$. From conditions \ref{cond:type:posconj} and \ref{cond:type:posdisj} of the definition of type (Definition \ref{def:type}), it follows that $\bigwedge \ell_\fr W^+(w) \not \in \ell_\fr X^-(y)$ and $\bigvee \ell_\fr W^-(w) \not \in \ell_\fr X^+(y)$. But then, condition \ref{cond:type:implication} gives us $\Sim w \not \in \ell_\fr X^+(y)$, so $\Sim w \not \in \ell_\fr X^+(x)$.

For the inductive step, by the same reasoning as in the base case it follows that $\bigwedge \ell_\fr W^+(w) \not \in \ell_\fr X^-(y)$ and $\bigvee \ell_\fr W^-(w) \not \in \ell_\fr X^+(y)$. Now, let $v$ be such that $v \succ w$. Since $\simrel$ is forward confluent then $v \simrel z_{v}$ for some $z_{v} \seq y$. By induction hypothesis, $\Sim v \not \in \ell^+(z_{v})$, so $\Sim v \not \in \ell^+(y)$. Since $v$ was arbitrary we conclude that $\bigvee _{v \succ w} \Sim v \not \in \ell^+(y)$. Finally, from condition \ref{cond:type:implication} of Definition \ref{def:type} and the fact that $y \peq x$ we get that $\Sim w \not \in \ell^+(x)$. 	
\endproof

\begin{remark}
Proposition \ref{propSimForm} more generally holds when $\mathfrak X$ is any labelled  space (not necessarily Aleksandroff), but this restricted version will suffice for our purposes.
\end{remark}

\section{The initial quasimodel}\label{seccan}

In this section we review the initial weak quasimodel $\irr \Sigma $ \cite{FernandezITLc} and use it to define an initial quasimodel $\fr J_\Sigma$. These structures are `initial' in the sense that if $\fr A$ is any labelled system, there exist surjective simulations from both $\irr \Sigma $ and $\fr J_ \Sigma $ to $\fr A$, i.e., they are initial in a category-theoretic sense.

\begin{theorem}\label{thmSurjI}
Given $\Sigma\Subset \fr \landif$, there exists a finite, saturated weak quasimodel $\irr\Sigma$
such that if $\fr A$ is any deterministic weak quasimodel then ${\rightharpoonup} \subseteq |\irr\Sigma| \times |\fr A|$ is a surjective dynamic simulation.
\end{theorem}

We do not need to elaborate on the construction of $\irr\Sigma$ here, but this is done in detail in \cite{FernandezITLc}.
Points of $\irr\Sigma$ are called {\em moments.}
One can think of $\irr\Sigma$ as a finite initial structure over the category of labelled weak quasimodels. 
Next, we will internalize the notion of simulating elements of $\irr\Sigma$ into the temporal language. This is achieved by the formulas $\Sim w$.



\begin{proposition}\label{propsub}
Let $\Lambda$ be a logic extending ${\sf ITL}^0_\nx$ over a temporal language $\cl L$.
Fix $\Sigma \Subset \cl L$ and let $\fr I = \irr\Sigma $, $w\in |\irr{}|$ and $\psi\in \Sigma $.
	\begin{multicols}{2}
\begin{enumerate}[label=\arabic*)]
	\item\label{itPropsubOne} \mbox{If $\psi\in \ell^-({{w}})$, then $\vdash \psi \to \mathrm{Sim}({{w}}) $.}
	\item\label{itPropsubOneb} If $\psi\in \ell^+({w})$, then\\ $\vdash \big (\psi \to \Sim {{w}} \big )\to \Sim w $.

	\item\label{itPropsubThree}\mbox{If ${{{w}}}\peq{{v}}$, then $\vdash \mathrm{Sim}({{{v}}})\to \mathrm{Sim}({{w}})$.}
	
	\columnbreak
	
	\item\label{itPropsubFour} \mbox{$\vdash \displaystyle \bigwedge_{\psi \in \ell_{\irr{}}^- ({{w}})}  \mathrm{Sim}({{w}}) \rightarrow \psi.$}
	\item\label{itPropsubFive} \mbox{$\vdash\displaystyle \tnext\bigwedge _{{{w}} \mathrel S_{\irr{}} {{{v}}}  }\mathrm{Sim}({{{v}}}) \to \mathrm{Sim}({{w}}).$}
\end{enumerate}
	\end{multicols}
\end{proposition}

\proof

\noindent \ref{itPropsubOne}
First assume that $\psi\in \ell^- ({w})$, and toward a contradiction that $\nvdash \psi \to \Sim{{w}}$. By the Lindenbaum lemma there is $\Gamma \in | \CMod | $ such that $\psi \to \Sim{{w}} \in \Gamma^-$. Thus for some $\Theta \seq_\CIcon \Gamma$ we have that $\psi \in \Theta^+$ and $\Sim{{w}} \in \Theta^-$. But then by Proposition \ref{propSimForm} we have that $(\fr W,{w}) \simu (\CMod,\Delta)$ for some $\Delta \seq_\CIcon \Theta$, so that $\psi\in \Delta^-$, and by monotonicity $\psi \in \Theta^-$, contradicting the consistency of $\Theta$.
\medskip

\noindent \ref{itPropsubOneb} If $\psi\in \ell^+ ({w})$, we proceed similarly. Assume toward a contradiction that\linebreak $\nvdash \big ( \psi \to \Sim{{w}} \big ) \to \Sim w$. Then, reasoning as above there is $\Theta \in |\CMod|$ such that $\psi \to \Sim{{w}} \in \Theta^+$ and $\Sim w \in \Theta^-$. From Proposition \ref{propSimForm} we see that there is $\Delta \seq_c \Theta$ such that $(\fr W,w) \simu (\CMod,\Delta)$, so that $\psi \in \Delta^+$ and, once again by Proposition \ref{propSimForm}, $\Sim w \in \Delta^-$. It follows that $\psi \to \Sim w \not \in \Delta^+$; but in view of upward persistence, this contradicts that $\psi \to \Sim{{w}} \in \Theta^+$.

\medskip


\noindent \ref{itPropsubThree} Suppose that ${v}\seq{w}$. Reasoning as above, it suffices to show that if $\Gamma \in |\CMod|$ is such that $\Sim{{w}} \in \Gamma ^-$, then also $\Sim{{v}} \in \Gamma ^-$. But if $\Sim{{w}} \in \Gamma ^-$, there is $\Theta \seq _\CIcon \Gamma$ such that $(\irr{},w) \simu (\CMod, \Theta)$. By forward confluence $(\irr{},v) \simu (\CMod, \Delta)$ for some $\Delta \seq_\CIcon \Theta$. Thus by Proposition \ref{propSimForm}, $\Sim {v} \in \Delta^-$ and by upwards persistence $\Sim{v} \in \Gamma^-$. Since $\Gamma \in |\CMod| $ was arbitrary, the claim follows.
\medskip

\noindent \ref{itPropsubFour}
We prove that if $\Gamma \in |\CMod|$ is such that
\begin{equation}\label{EqGammaConjunc}
\bigwedge _{\psi\in \ell^- ({{w}})} \mathrm{Sim}({{w}})  \in \Gamma ^+,
\end{equation}
then $\psi \in \Gamma^+$. If \eqref{EqGammaConjunc} holds then by Theorem \ref{thmSurjI}, there is ${w} \in |\irr{}|$ with $(\irr{}, w) \simu (\CMod,\Gamma)$. By Proposition \ref{propSimForm}, $\Sim {w} \in \Gamma^-$, hence it follows from \eqref{EqGammaConjunc} that $\psi \not \in \ell^- ({{w}})$; but $\irr{}$ is saturated and $\psi \in \Sigma$, so $\psi  \in \ell^+ ({{w}})$ and thus $\psi \in \Gamma^+$, as required.
\medskip

\noindent \ref{itPropsubFive} Suppose that $\Gamma \in |\CMod|$ is such that
\begin{equation}\label{EqGammaCirc}
\tnext\bigwedge _{{{w}} \mathrel S_{\irr{}} {{{v}}}   }\Sim{{{v}}} \in \Gamma^+ ,
\end{equation}
and assume toward a contradiction that $\Sim{{w}} \in \Gamma^-$. By Proposition \ref{propSimForm} $(\irr{}, w) \simu (\CMod,\Delta)$ for some $\Delta \seq_\CIcon \Gamma$. Since $\simu$ is a dynamic simulation, it follows that there is $v \in |\irr{}|$ with ${w} \mathrel S_{\irr{}} v$ and $(\irr{}, v) \simu  \big ( \CMod, S_\CIcon (\Delta) \big ) $, so that $\mathrm{Sim}(v) \in \big ( S_\CIcon (\Delta) \big ) ^-$. It follows that $\tnext \mathrm{Sim}(v) \in \Gamma^-$, since $S_\CIcon$ is sensible and $\Gamma$ is saturated. But $\Delta\seq_\CIcon \Gamma$, so that $\tnext \Sim v \in S^-_\CIcon (\Gamma)$ as well, contradicting \eqref{EqGammaCirc}. 
\endproof

We are now ready to define our initial quasimodel. Given a finite set $\Sigma$ of formulas, we will define a quasimodel $\cqm\Sigma$ falsifying all unprovable $\Sigma$-types. This quasimodel is a substructure of $\irr\Sigma$, containing only moments which are {\em possible} in the following sense.

\begin{definition}\label{defsound}
Let $\Lambda$ be an admissible temporal logic over $\cl L \subseteq \landif$.
Fix $\Sigma\Subset \cl L$. We say that a moment ${{w}} \in |\irr\Sigma|$ is {\em possible} if $\not \vdash \mathrm{Sim}({{w}})$, and denote the set of possible $\Sigma$-moments by $\unp\Sigma$.
\end{definition}

\shortproof{The following gives an alternative characterization of Definition \ref{defsound} and can be checked using Proposition \ref{propSimForm} and the Lindenbaum lemma.
}

\begin{lemma}\label{lemPosTotal}
Let $\Lambda$ be an admissible temporal logic over $\cl L \subseteq \landif$ and $\Sigma\Subset \cl L$.
Then, ${{w}} \in |\irr\Sigma|$ is possible if and only if there is $\Gamma \in |\CMod|$ such that $( \irr\Sigma,  w ) \simu ( \CMod, \Gamma) $.
\end{lemma}

\fullproof{
\proof
If $(\irr\Sigma,w) \simu (\CMod,\Gamma)$ then by Proposition \ref{propSimForm} $\Sim w  \not\in \Gamma^+ $, hence since $\Gamma$ is saturated $\Sim w  \in \Gamma^-$ which since $\Gamma$ is consistent implies that $\Sim w $ is possible.
Conversely, if $\not \vdash \Sim w$ we can find by the Lindenbaum lemma $\Gamma \in |\CMod|$ with $\Sim w \in \Gamma^-$, which implies that $(\irr\Sigma,w) \simu (\CMod,\Gamma)$.
\endproof
}

With this we are ready to define our initial structure, which as we will see later is indeed a quasimodel.

\begin{definition}
Let $\Lambda$ be an admissible temporal logic over $\cl L \subseteq \landif$.
Given $\Sigma \Subset\cl L$, we define the {\em initial structure} for $\Sigma$ by $\cqm \Sigma = \irr {\Sigma} \upharpoonright \unp\Sigma$.
\end{definition}

\begin{remark}
In principle $\cqm \Sigma$ depends on $\Lambda$, but we do not reflect this in the notation since $\Lambda$ will always be either ${\sf ITL}^0_\ps$ or ${\sf ITL}^0_{\ps\forall}$, depending on whether $\forall$ appears in $\Sigma$.
\end{remark}

Our strategy from here on will be to show that canonical structures are indeed quasimodels; once we establish this, completeness of ${\logbasic} $ is an easy consequence. The most involved step will be showing that the successor relation on ${\cqm{\Sigma}}$ is $\omega$-sensible, but we begin with some simpler properties.

\begin{lemma}\label{niceprop}
Let $\Lambda$ be a logic extending ${\sf ITL}^0_\nx$ over a temporal language $\cl L \subseteq \landif$.
Let $\Sigma \Subset \cl L$, $\irr{} = \irr\Sigma$ and $\cqm{} = \cqm \Sigma$. Then, $|\cqm{}|$ is an open subset of $|\irr{}|$ and $S_{\cqm{}}$ is serial.
\end{lemma}

\proof
To check that $|\cqm{}|$ is upward closed, let ${{w}}\in |\cqm{}|$ and suppose ${{{v}}}\seq{{w}}$. Now, by Proposition \ref{propsub}.\ref{itPropsubThree}, we have that $\vdash\mathrm{Sim}({{{v}}})\to \mathrm{Sim}({{w}})$; hence if ${{w}}$ is possible, so is ${{{v}}}$.
To see that $S_{\cqm {} }$ is serial, observe that by Proposition \ref{propsub}.\ref{itPropsubFive}, if  ${{w}}\in |\cqm{}| \subseteq |\irr{}|$, $\vdash \tnext\bigwedge_{{{w}}\mathrel S_{\irr{}} {{{v}}} }\mathrm{Sim}({{{v}}}) \to \mathrm{Sim}({{w}})$. Since ${{w}}$ is possible, it follows that for some ${{{v}}}$ with ${{w}} \mathrel S_{\irr{}}  {{{v}}}$, ${{{v}}}$ is possible as well, for otherwise $\tnext\bigwedge_{{{w}}\mathrel S_{\irr{}} {{{v}}} }\mathrm{Sim}({{{v}}})$ would be equivalent to $\tnext \top$, allowing us to deduce $\mathrm{Sim}({{w}})$. But then ${{{v}}}\in|\cqm{}|$, as needed.
\endproof

\section{$\omega$-Sensibility}\label{secOmSens}
In this section we will show that $S_{\cqm{}}$ is $\omega$-sensible, the most difficult step in proving that $\cqm{}=\cqm{\Sigma}$ is a quasimodel.
Fix an admissible temporal logic $\Lambda$ over $\cl L \subseteq \landif$, and let $R$ denote the transitive, reflexive closure of $S_{\cqm{}}$.
If $w \mathrel R v$, we say that $v$ is {\em reachable} from $w$.

\begin{lemma}\label{syntactic}
Let $\Lambda$ be an admissible temporal logic over $\cl L \subseteq \landif$.
If $\Sigma\Subset \cl L$ and ${{w}}\in|\cqm\Sigma|$, then  $\vdash \tnext \bigwedge\limits_{w \mathrel R v}\mathrm{Sim}({{{v}}})\to \bigwedge \limits_{w \mathrel R v} \mathrm{Sim}({{{v}}})$.
\end{lemma}

\proof
Let $\irr{} = \irr\Sigma$. By Proposition \ref{propsub}.\ref{itPropsubFive} we have that, for all ${{{v}}}\in{R}({{w}})$,
\[\vdash \tnext\bigwedge\limits_{{{{v}}} \mathrel S_{\irr{}} u }\mathrm{Sim}(u) \to \mathrm{Sim}({{{v}}}).\]
Now, if $u\not\in|\cqm{\Sigma}|$, then $\vdash \mathrm{Sim}(u)$, hence by \ref{ax14NecCirc} we have
$\vdash \tnext \mathrm{Sim}(u)$, and we can remove $\Sim u$ from the conjunction using Lemma \ref{lemmReverseNext} and propositional reasoning.
Since ${{{v}}} \in R(w)$ was arbitrary, this shows that
\[\vdash \tnext \bigwedge\limits_{w \mathrel R v}\mathrm{Sim}({{{v}}})\to \bigwedge \limits_{w \mathrel R v}\mathrm{Sim}({{{v}}}).\]
\endproof

From this we obtain the following, which evidently implies $\omega$-sensibility:

\begin{proposition}\label{tempinc}
Let $\Lambda$ be a logic extending ${\sf ITL}^0_\ps$ over a temporal language $\cl L \subseteq \landif$ and $\Sigma\Subset \cl L$.
If ${{w}}\in|\cqm\Sigma|$ and $\diam \psi\in \ell^+ ({{w}})$, then there is ${{{v}}}\in{R}({{w}})$ such that $\psi\in \ell^+ ({{{v}}})$.
\end{proposition}

\proof Towards a contradiction, assume that ${{w}}\in \unp \Sigma$ and $\diam \psi\in \ell^+ ({{w}})$ but, for all ${{{v}}}\in{R}({{w}})$, $\psi \in \ell^-({{w}})$.
By Lemma \ref{syntactic}, $\vdash \tnext \bigwedge \limits_{w \mathrel R v} \mathrm{Sim}({{{v}}})\to \bigwedge\limits_{w \mathrel R v} \mathrm{Sim}({{{v}}})$.
By the $\diam$-induction rule \ref{ax12:ind:2}, $\vdash \diam \bigwedge \limits_{w\mathrel Rv}\mathrm{Sim}({{{v}}})\to \bigwedge\limits_{w\mathrel Rv}\mathrm{Sim}({{{v}}})$;
in particular,
\begin{equation}\label{other}
\vdash \diam \bigwedge _{w\mathrel Rv}\mathrm{Sim}({{{v}}})\to \mathrm{Sim}({{w}}).
\end{equation}

Now let ${{{v}}}\in{R}({{w}})$. By Proposition \ref{propsub}.\ref{itPropsubOne} and the assumption that $\psi \in \ell^-({{{v}}})$ we have that
$\vdash \psi \to \mathrm{Sim}({{{v}}}) $,
and since ${{{v}}}$ was arbitrary,
$\vdash \psi \to \bigwedge_{w\mathrel Rv}\mathrm{Sim}({{{v}}}) $.
Using distributivity \ref{ax11:dist} we further have that
$\vdash \diam \psi \rightarrow \diam  \bigwedge_{w\mathrel Rv}\mathrm{Sim}({{{v}}})$.
This, along with (\ref{other}), shows that
$\vdash \diam \psi \to \mathrm{Sim}({{w}})$;
however, by Proposition \ref{propsub}.\ref{itPropsubOneb} and our assumption that $\diam \psi\in \ell^+ ({{w}})$ we have that
$\vdash \big ( \diam \psi \to \mathrm{Sim}({{w}}) \big ) \to \Sim w$,
hence by modus ponens we obtain $\vdash \mathrm{Sim}({{w}}),$ which contradicts the assumption that ${{w}}\in \unp\Sigma$. We conclude that there can be no such ${{w}}$.
\endproof

\begin{corollary}\label{laststretch}
Let $\Lambda$ be an admissible temporal logic over $\cl L \subseteq \landif$.
Then, if $\Sigma\Subset \cl L$, $\cqm\Sigma$ is a quasimodel.
\end{corollary}

\proof
Let $\cqm{} = \cqm \Sigma$. By Lemma \ref{niceprop}, $|\cqm{}|$ is upwards closed in $|\irr\Sigma|$ and $S_{\cqm{}}$ is serial, while by Proposition \ref{tempinc}, $S_{\cqm{}}$ is $\omega$-sensible. It follows from Lemma \ref{LemmIsQuasi} that $\cqm{}$ is a quasimodel.
\endproof

We are now ready to prove that ${\logbasic} $ is complete.

\begin{theorem}\label{theocomp}
If $\varphi \in \landi$ is valid over the class of dynamical systems, ${\logbasic} \vdash\varphi$.
\end{theorem}

\proof
We prove the contrapositive. Suppose $\varphi$ is an unprovable formula and let
\[W=\cbra {{w}}\in\irr{{\rm sub}(\varphi)}:\varphi\in \ell^- ({{w}})\cket.\]
Then, by Proposition \ref{propsub}.\ref{itPropsubFour} we have that
$\vdash \bigwedge_{{w} \in W} \mathrm{Sim}({{w}}) \rightarrow \varphi$;
since $\varphi$ is unprovable, it follows that some ${{w}}^\ast\in W$ is possible and hence ${{w}}^\ast\in \unp {{\rm sub}(\varphi)}$. By Corollary \ref{laststretch}, $\cqm {{\rm sub}(\varphi)}$ is a quasimodel, so that by Theorem \ref{TheoITLc}, $\varphi$ is falsifiable in some dynamical system.
\endproof

\section{The universal modality}\label{secUniversal}

Now let us show that ${\sf ITL}^0_{\ps\forall}$ is complete for the class of dynamical systems.
As before our completeness proof relies on the canonical model $\mathcal M_{{\sf ITL}^ 0_{\ps\forall}}$ and the initial quasimodel $\cqm {\Sigma}$ (for suitable $\Sigma$), but now we cannot use these structures as they are as they are not honest (see Definition \ref{frame}).
We will first exhibit an honest substructure of $\mathcal M_{{\sf ITL}^ 0_{\ps\forall}}$.

\begin{definition}
Let $\Sigma \subseteq \landif$. We define $\Sigma_\forall$ to be the set of formulas of $\Sigma$ of the form $\forall \varphi$.
A {\em universal $\Sigma$-profile} is a partition $\Pi = (\Pi^+,\Pi^-)$ of $\Sigma_\forall$.
If $\Phi = (\Phi^+,\Phi^-)$ is a pair of sets of formulas, we define $\Phi_\forall = \big ((\Phi^+)_\forall , (\Phi^-)_\forall \big )$, which we will henceforth write as $ ( \Phi^+_\forall , \Phi^-_\forall )$.
\end{definition}


\begin{definition}
Given $\Sigma \subseteq \landif$, a $\Sigma$-labelled structure $\fr A$ and a universal $\Sigma$-profile $\Pi$, we define $\fr A [ \Pi ] = \fr A \upharpoonright \{w\in |\fr A| : \Pi\subseteq_T \ell_\forall (w) \}$.
\end{definition}

\begin{lemma}\label{lemUniversalRestrict}
Let $\Sigma \Subset \landif$, $\Lambda = {\sf ITL}^0_{\ps \forall}$, $\CMod$ be the canonical $\Lambda$-model and $\Pi$ a universal $\Sigma$--profile. Then, $|\CMod[\Pi]|$ is open and $\CMod[\Pi]$ is honest as a $\Sigma$-labelled quasimodel.
\end{lemma}

\fullproof{\proof
First we check that $|\CMod[\Pi]|$ is open, i.e. if $\Phi \peq_\CIcon \Psi$ and $\Phi \in |\CMod[\Pi]|$ then $\Psi \in |\CMod[\Pi]|$.
Let $\forall \varphi \in \Pi^+$, so that $\forall \varphi \in \Phi^+$; by monotonicity, $\forall \varphi \in \Pi^+$.
If instead $\forall \varphi \in \Pi^-$, \ref{axUnivEM} yields $\neg \forall \varphi \in \Phi^+$ and hence $\forall \varphi \in \Psi^-$.
We conclude that $\Psi \in  |\CMod[\Pi]|$.
In a similar fashion we see using \ref{axUnivNex} that if $\Phi \in  |\CMod[\Pi]|$ then $S_\CIcon (\Phi) \in |\CMod[\Pi]|$. In view of Lemma \ref{LemIsQuasi}, we may moreover conclude that $\CMod[\Pi]$ is a weak quasimodel.

It remains to show that $\CMod[\Pi]$ is honest. Let us take $\Psi \in |\CMod[\Pi]|$  and let $\forall \varphi \in \Sigma$ be such that $\forall \varphi \in \ell^+(\Psi) $ (which is equal to $ \Psi^+$). By definition, $\Psi \in |\CMod|$ and $\Pi \subseteq_T \Psi_\forall$. Let us take  $\Psi \in |\CMod[\Pi]|$, so $v \in |\CMod|$ and $\Pi \subseteq_T \Psi_\forall$. Consequently, $\forall \varphi \in \Psi^+$. Finally, thanks to Axiom~\ref{axUnivT} we can conclude that $\varphi \in \Psi^+$. Since $\Psi$ is generic, we conclude that $\CMod[\Pi]$  satisfies the first condition of honesty.

For the second condition, let us take $ \Phi \in |\CMod[\Pi]|$  and let $\forall \varphi \in \Sigma$ be such that $\forall \varphi \in \Phi^- $. 
Let $\Psi' = ( \Pi^+ , \lbrace \varphi\rbrace \cup  \Pi^-)$.
Assume towards a contradiction that $\Psi'$ is not consistent, so that $\vdash \bigwedge \Pi^+ \rightarrow \varphi \vee \bigvee \Pi^-  $.
By Rule~\ref{rulUnivNec}, $\vdash \forall \left(\bigwedge \Pi^+ \rightarrow \varphi \vee \bigvee \Pi^- \right)$.
By Axiom~\ref{axUnivK}, $\vdash \forall \bigwedge \Pi^+  \rightarrow \forall\left( \varphi \vee \bigvee \Pi^- \right)$.
From $\Phi_\forall = \Pi$ and axioms~\ref{axUniv4},~\ref{axUnivK} and propositional reasoning we conclude that $\vdash \bigwedge \Pi^+ \to \forall \bigwedge \Pi^+ $, hence $\vdash \bigwedge \Pi^+ \to \forall\left(\varphi \vee \bigvee \Pi^-  \right) $.
By several applications of Axiom~\ref{axUnivVee}, $\vdash \forall \left ( \varphi \vee \bigvee \Pi^-\right)  \to \forall \varphi \vee \bigvee \Pi^-$, hence $\vdash \bigwedge \Pi^+ \to \forall \varphi \vee \bigvee \Pi^-$, which since $\forall \varphi \in \Pi ^- \subseteq \Phi^- $ implies that $\Phi$ is inconsistent, a condtradiction.

Hence $\Psi'$ is consistent. By Lemma~\ref{LemmLind}, $\Psi'$ can be extended to a prime type $
\Psi$.
Since, by definition, $\Pi \subseteq_T \Psi$, we have that $\Psi \in |\CMod[\Pi]|$ and, moreover, $\varphi \in \Psi^-$, as required.
\endproof}

\shortproof{\proof
It can be checked using universal excluded middle \ref{axUnivEM} that $|\CMod[\Pi]|$ is open and using \ref{axUnivNex} that it is $S_\CIcon$-invariant.
In view of Lemma \ref{LemIsQuasi}, we may moreover conclude that $\CMod[\Pi]$ is a weak quasimodel.

It remains to show that $\CMod[\Pi]$ is honest. That $\forall \varphi \in \Phi^+ \cap \Sigma $ implies that $\varphi \in \Psi^+$ for all $\Phi,\Psi \in |\CMod[\Pi]|$ follows readily from the truth axiom \ref{axUnivT}.
For the remaining condition, let us take $ \Phi \in |\CMod[\Pi]|$  and let $\forall \varphi \in \Sigma$ be such that $\forall \varphi \in \Phi^- $. 
Let $\Psi' = ( \Pi^+ , \lbrace \varphi\rbrace \cup  \Pi^-)$.
Assume towards a contradiction that $\Psi'$ is not consistent, so that $\vdash \bigwedge \Pi^+ \rightarrow \varphi \vee \bigvee \Pi^-  $.
By Rule~\ref{rulUnivNec}, $\vdash \forall \left(\bigwedge \Pi^+ \rightarrow \varphi \vee \bigvee \Pi^- \right)$.
By Axiom~\ref{axUnivK}, $\vdash \forall \bigwedge \Pi^+  \rightarrow \forall\left( \varphi \vee \bigvee \Pi^- \right)$.
From $\Phi_\forall = \Pi$ and axioms~\ref{axUniv4},~\ref{axUnivK} and propositional reasoning we conclude that $\vdash \bigwedge \Pi^+ \to \forall \bigwedge \Pi^+ $, hence $\vdash \bigwedge \Pi^+ \to \forall\left(\varphi \vee \bigvee \Pi^-  \right) $.
By several applications of Axiom~\ref{axUnivVee}, $\vdash \forall \left ( \varphi \vee \bigvee \Pi^-\right)  \to \forall \varphi \vee \bigvee \Pi^-$, hence $\vdash \bigwedge \Pi^+ \to \forall \varphi \vee \bigvee \Pi^-$, which since $\forall \varphi \in \Pi ^- \subseteq \Phi^- $ implies that $\Phi$ is inconsistent, a condtradiction.

Hence $\Psi'$ is consistent. By Lemma~\ref{LemmLind}, $\Psi'$ can be extended to a prime type $
\Psi$.
Since, by construction, $\Pi \subseteq_T \Psi$, we have that $\Psi \in |\CMod[\Pi]|$ and, moreover, $\varphi \in \Psi^-$, as required.
\endproof}

This already is sufficient to prove that our logics over $\cl L_{\nx\forall}$ are complete.
The proof of the following is analogous to that of Theorem \ref{theoCircComp}, but using the structures $\CMod [\Pi]$ instead of $\CMod $.

\begin{theorem}\label{theoCircUnivComp}
${\sf ITL}^0_{\nx\forall}$ is complete for the class of expanding posets and $\logpersax_{\nx\forall}$ for the class of persistent posets.
\end{theorem}

\begin{remark}
We will not go into detail regarding strong completeness in this article, but Theorems \ref{theoCircComp} and \ref{theoCircUnivComp} can be strengthened to state that these logics are strongly complete.
Note that logics with $\diam$ cannot be strongly complete since they are not compact.
\end{remark}

For the language with $\diam$ we will use an honest substructure of $\irr \Sigma$, for which we use the following result of \cite{FernandezITLc}.

\begin{lemma}\label{LemTotSim}
Suppose that $\Sigma\subseteq \Delta \subseteq \landif$ are both closed under subformulas, $\mathfrak X$ is a $\Sigma$-labelled space, $\mathfrak Y$ is a $\Delta$-labelled space, and $\chi\subseteq |\mathfrak X|\times |\mathfrak Y|$
is a total, surjective simulation. Then, if $\ell_\mathfrak Y$ is honest, it follows that $\ell_\mathfrak X$ is honest as well.
\end{lemma}

Finally, we observe that $|\cqm \Sigma [\Pi]|$ is a quasimodel.

\begin{lemma}\label{lemSInv}
If $\Sigma\Subset \landif$ and $\Pi$ is a $\Sigma$-universal profile then ${\simu} \cap {| \cqm \Sigma[\Pi]| \times |\CMod [\Pi]|}$ is total and surjective and $\cqm \Sigma[\Pi]$ is a quasimodel.
\end{lemma}

\proof
Let $\cqm{} = \cqm\Sigma$.
If $ w \in |\cqm {}[\Pi]|$, by Lemma \ref{lemPosTotal} there is $\Gamma \in |\CMod|$ with $(\cqm {}, w ) \simu (\CMod,\Gamma)$, and by label-preservation $\Pi \subseteq_T \Gamma$, so that $\Gamma \in |\CMod[\Pi]|$.
Hence $\simu$ is total.
Conversely, if $\Gamma \in | \CMod[\Pi]|$ then by Theorem \ref{thmSurjI} there is $w\in |\cqm {} [\Pi]|$ such that $(\cqm {}, w ) \simu (\CMod,\Gamma)$, and once again by label-preservation $w \in |\cqm{} [\Pi]|$.

To see that $\cqm{} [\Pi]$ is a quasimodel, in view of Lemma \ref{LemmIsQuasi}, it suffices to show that $|\cqm {}[\Pi]|$ is open and $S_{\cqm{}}$-invariant.
However, if $ w \in |\cqm {}[\Pi]|$, then since $ \simu$ is total we have that there is $\Gamma \in |\CMod[\Pi]|$ with $(\cqm {}, w ) \simu (\CMod,\Gamma)$.
By Lemma \ref{lemUniversalRestrict} $|\CMod[\Pi]|$ is open, hence by continuity $v\in {\simu^{-1}}(|\CMod[\Pi]|)$, and once again by label-preservation this shows that $\Pi\subseteq_T \ell(v) $, so that $v\in |\cqm {}[\Pi]|$.
That $|\cqm {}[\Pi]|$ is $S_{\cqm {}}$-invariant is immediate from conditions \ref{ItCompFive} and \ref{ItCompSix} of Definition \ref{compatible} and the fact that $S_{\cqm {}}$ is sensible.
\endproof

\begin{theorem}\label{theoUnivComp}
The logic ${\sf ITL}^0_{\ps\forall}$ is complete for the class of dynamical systems.
\end{theorem}

\proof
If $\varphi$ is unprovable, then $\varphi \in \Phi^-$ for some $\Phi \in |\CMod|$, and by Theorem \ref{thmSurjI}, there is some $w\in |\cqm \Sigma|$ such that $(\cqm \Sigma, w) \rightharpoonup (\CMod, \Phi)$. Let $\Pi = (\ell(w))_\forall$. Then, $\CMod[\Pi]$ is an honest weak quasimodel, so that by Lemmas \ref{LemTotSim} and \ref{lemSInv}, so is $\cqm \Sigma[\Pi]$.
It follows by Theorem \ref{TheoITLc} that $\varphi$ is falsifiable on some dynamic topological model.
\endproof

\section{Completeness for expanding posets}\label{SecCons}

Our goal for this section is to show that the temporal logics of dynamic posets and of dynamical systems coincide with respect to $\landi$.
We will show this by `unwinding' a quasimodel to produce a dynamical poset. First we discuss some operations on types that will be used in the unwinding.
If $\Sigma$ is a set of formulas, first define $\posres\Psi \Sigma = (\Psi^+\cap \Sigma , \Psi^-)$, and ${\rm sub} (\Sigma) = \bigcup _{\varphi \in \Sigma}{\rm sub}(\varphi)$. With this, we have the following:

\begin{lemma}\label{LemmRestrict}
Let $\Phi,\Psi,\Gamma,\Theta$ be $\landi$-types and $\Sigma \subseteq \landi$ closed under subformulas. Then,
\begin{enumerate}

\item\label{cond:alsotype} $\posres\Phi \Sigma$ is also a type;

\item\label{ItCrossTrans} if $\Gamma \sqsubT \Phi \peqT \Psi$ or $\Gamma \peqT \Phi \sqsubT \Psi$ then $\Gamma \peqT \Psi$,\david{Maybe define $\sqsubT$ in this section, as it is not used before.} and


\item\label{ItRestrictSqsub} if $\Gamma \sqsubT \Phi \ST \Psi$ and ${\rm sub} ( \Gamma^+ ) \subseteq \Sigma$, then $\Gamma \ST \posres \Psi \Sigma$.


\end{enumerate}
\end{lemma}

\shortproof{
\proof
	
To prove item \ref{cond:alsotype} it is sufficient to check that the conditions of Definition~\ref{def:type} hold. 
Conditions \ref{cond:type:intersection} and~\ref{cond:type:bot} of Definition~\ref{def:type} are straightforward.
Since $\Phi^- = (\posres\Phi \Sigma)^-$, conditions \ref{cond:type:negconj} and~\ref{cond:type:negdisj} clearly hold.
For condition~\ref{cond:type:implication}, suppose that $\varphi\to\psi\in  (\posres\Psi \Sigma)^+ $.
Since $\Sigma$ is closed under subformulas, $\varphi, \psi \in \Sigma$ and, since $\Psi$ is a type it follows that either $\varphi \in \Psi^-$ or $\psi \in \Psi^+$.
By definition either $\varphi \in \Psi^-$ or $\psi \in \Psi^+ \cap \Sigma$.
The proofs for conditions \ref{cond:type:posconj} and~\ref{cond:type:posdisj} of Definition~\ref{def:type} are similar and left to the reader. 

Regarding item~\ref{ItCrossTrans} of the lemma, on one side, $\Gamma \sqsubT \Phi \peqT \Psi$ means that $\Gamma^+ \subseteq \Phi^+ \subseteq \Psi^+$ and $\Gamma^- = \Phi^- \supseteq \Psi^-$. Therefore $\Gamma^+ \subseteq \Psi^+$ and $\Psi^- \subseteq \Gamma^-$ so $\Gamma \peqT \Psi$. On the other side  $\Gamma \peqT \Phi \sqsubT \Psi$ means by definition that $\Gamma^+ \subseteq \Psi^+ \subseteq \Psi^+$ and $\Gamma^- \supseteq \Phi^- = \Psi^-$. It follows that $\Gamma^+ \subseteq \Psi^+$ and $\Psi^+ \subseteq \Gamma^-$ so $\Gamma \peqT \Psi$.

For item~\ref{ItRestrictSqsub} we must check that each condition of Definition~\ref{compatible} holds. As an example, we work out \ref{ItCompThree}. If $\diam \psi \in \Gamma^+$, since ${\rm sub}(\Gamma^+) \subseteq \Sigma$ then $\diam \psi, \psi \in \Sigma$. From $\Gamma \sqsubT  \Phi\ST\Psi$ we conclude that $\diam \psi \in \Phi^+$ and either $\psi \in \Gamma^+$ or $\diam \psi \in \Psi^+$. From this it follows that either $\psi \in \Gamma^+$ or  $\diam \psi \in \Psi^+ \cap \Sigma$ (which means that $\diam \psi \in \left(\posres \Psi \Sigma\right)^+$).
Other conditions follow similar reasoning and are left to the reader.
\endproof
}

\fullproof{
\proof
	
To prove item \ref{cond:alsotype} it is sufficient to check that the conditions of Definition~\ref{def:type} hold. 
Conditions \ref{cond:type:intersection} and~\ref{cond:type:bot} of Definition~\ref{def:type} are straightforward.
Since $\Phi^- = (\posres\Phi \Sigma)^-$, conditions \ref{cond:type:negconj} and~\ref{cond:type:negdisj} clearly hold.
For condition~\ref{cond:type:implication}, suppose that $\varphi\to\psi\in  (\posres\Psi \Sigma)^+ $.
Since $\Sigma$ is closed under subformulas, $\varphi, \psi \in \Sigma$ and, since $\Psi$ is a type it follows that either $\varphi \in \Psi^-$ or $\psi \in \Psi^+$.
By definition either $\varphi \in \Psi^-$ or $\psi \in \Psi^+ \cap \Sigma$.
The proofs for conditions \ref{cond:type:posconj} and~\ref{cond:type:posdisj} of Definition~\ref{def:type} are similar and left to the reader. 

Regarding item~\ref{ItCrossTrans} of the lemma, on one side, $\Gamma \sqsubT \Phi \peqT \Psi$ means that $\Gamma^+ \subseteq \Phi^+ \subseteq \Psi^+$ and $\Gamma^- = \Phi^- \supseteq \Psi^-$. Therefore $\Gamma^+ \subseteq \Psi^+$ and $\Psi^- \subseteq \Gamma^-$ so $\Gamma \peqT \Psi$. On the other side  $\Gamma \peqT \Phi \sqsubT \Psi$ means by definition that $\Gamma^+ \subseteq \Psi^+ \subseteq \Psi^+$ and $\Gamma^- \supseteq \Phi^- = \Psi^-$. It follows that $\Gamma^+ \subseteq \Psi^+$ and $\Psi^+ \subseteq \Gamma^-$ so $\Gamma \peqT \Psi$.

For item~\ref{ItRestrictSqsub} we consider the conditions of Definition~\ref{compatible}:\medskip

\noindent \ref{ItCompOne}. If $\tnext \psi \in \Gamma^+$, from $\Gamma \sqsubT  \Phi\ST\Psi$ we conclude that $\tnext \psi \in \Phi^+$ and $\psi \in \Psi^+$. Since ${\rm sub}(\Gamma^+) \subseteq \Sigma$ then $\psi \in \Sigma$. Therefore $\psi \in \Psi^+ \cap \Sigma$ so $\psi \in \left(\posres\Psi \Sigma\right)^+$.\medskip

\noindent \ref{ItCompTwo}. If $\tnext \psi \in \Gamma^-$, from $\Gamma \sqsubT  \Phi\ST\Psi$ we conclude that $\tnext \psi \in \Phi^-$ and $\psi \in \Psi^-$, which by definition means that $\psi \in \left(\posres\Psi \Sigma  \right)^-$.
\medskip

\noindent \ref{ItCompThree}. If $\diam \psi \in \Gamma^+$, since ${\rm sub}(\Gamma^+) \subseteq \Sigma$ then $\diam \psi, \psi \in \Sigma$. From $\Gamma \sqsubT  \Phi\ST\Psi$ we conclude that $\diam \psi \in \Phi^+$ and either $\psi \in \Gamma^+$ or $\diam \psi \in \Psi^+$. From this it follows that either $\psi \in \Gamma^+$ or  $\diam \psi \in \Psi^+ \cap \Sigma$ (which means that $\diam \psi \in \left(\posres \Psi \Sigma\right)^+$).
\medskip

\noindent \ref{ItCompFour}. If $\diam \psi \in \Gamma^-$, from $\Gamma \sqsubT  \Phi\ST\Psi$ we conclude that $\diam \psi \in \Phi^-$, $\psi \in \Phi^-$ (thus $\psi \in \Gamma^-$) and $\diam \psi \in \Psi^-$. As a consequence it follows that $\psi \in \Gamma^-$ and $\diam \psi \in \left(\posres\Psi \Sigma\right)^-$.
\endproof
}

We may also wish to `forget' temporal formulas that have been realized.
To make this precise, let ${\rm sup}(\varphi)$ denote the set of {\em super}-formulas of $\varphi$, i.e., ${\rm sup}(\varphi)=\{\psi \in \landi : \varphi \in {\rm sub}(\psi)\}$.
Say that a formula $\varphi$ is a {\em temporal formula} if it is of the forms $\tnext \psi$ or $\diam \psi$, and if $\Phi$ is a set of formulas, say that $\varphi \in \Phi$ is {\em maximal in $\Phi$} if it does not have any temporal superformulas in $\Phi$ (except $\varphi$).
Then, define $\remove\Phi\varphi =(\Phi^+\setminus {\rm sup}(\varphi),\Phi^-)$.

\begin{lemma}\label{LemDelete}
	Suppose that $\Phi \ST \Psi$.
	\begin{enumerate*}[label=\arabic*)]		
		\item\label{ItDeleteCirc}  If $\tnext\varphi$ is maximal in $\Phi^+$, then $\Phi \ST ( \remove \Psi {\tnext\varphi})$ .		
		\item\label{ItDeleteDiam}  If $\diam \varphi $ is maximal in $\Phi^+$ and $\varphi \in \Phi^+$, then $\Phi \ST ( \remove \Psi {\diam \varphi})$.		
	\end{enumerate*}
\end{lemma}

\begin{proof}[Proof sketch]
We consider the first item; the second is analogous. Assuming that $\tnext\varphi$ is maximal in $\Phi^+$,
it must be checked that the four conditions of Definition~\ref{compatible} hold.
For conditions \eqref{ItCompOne} and~\eqref{ItCompThree}, remark that ${\rm sup}(\tnext \varphi) \cap \Phi^+ = \{\tnext\varphi\}$.
Therefore if $\tnext \theta$ or $\diam \theta$ belong to $\Phi^+$, then neither $\theta$ nor $\diam \theta$ belong to ${\rm sup}(\tnext \varphi)$.
For conditions \eqref{ItCompTwo} and~\eqref{ItCompFour}, it suffices to observe that $(\Psi \setminus \tnext\varphi)^- = \Psi^-$.
\ignore{
let us check that the four conditions of Definition~\ref{compatible} hold.
\medskip

\noindent\eqref{ItCompOne} If $\tnext \theta \in \Phi^+$, since $\Phi \ST \Psi$ then $\theta \in \Psi^+$. Moreover, since $\tnext \theta \in \Phi^+$ then $\theta \not \in {\rm sup}(\tnext \varphi)$ by maximality of $\tnext\varphi$. Therefore, $\theta \in \left(\Psi^+\setminus {\rm sup}(\tnext \varphi)\right) =\left(\remove \Psi {\tnext\varphi} \right)^+$.\medskip

\noindent\eqref{ItCompTwo} If $\tnext \theta \in \Phi^-$, since $\Phi \ST \Psi$ then $\theta \in \Psi^- = \left(\remove \Psi {\tnext\varphi}\right)^-$.\medskip

\noindent\eqref{ItCompThree} If $\diam \theta \in \Phi^+$, since $\Phi \ST \Psi$ then $\theta \in \Phi^+$ or $\diam \theta \in \Psi^+$.
In the latter case, since $\tnext\varphi$ is maximal in $\Phi^+$ and $\diam \theta \in \Phi^+$, it follows that $\diam \theta \not \in {\rm sup}(\tnext \varphi)$.
Therefore, $\diam \theta \in =\left(\remove \Psi {\tnext\varphi} \right)^+$.
\medskip

\noindent\eqref{ItCompFour} If $\diam \theta \in \Phi^-$, since $\Phi \ST \Psi$ then $ \diam \theta \in \left(\remove \Psi {\tnext\varphi}\right)^-$.
}
\end{proof}

The unwinding procedure is similar to that in \cite{FernandezITLc}. There, the points of the `limit model' obtained from a quasimodel are the infinite paths satisfying all $\diam$-formulas in their labels. However, to obtain a poset rather than a topological space, we will instead work with finite paths.

\begin{definition}\label{def:term-path}
If $\fr Q$ is an $\landi$-quasimodel, a {\em path (on $\fr Q$)} is a sequence $(w_i)_{i<n} \subseteq |\fr Q|$ such that $w_i \mathrel S w_{i+1}$ for all $i<n-1$. We define a {\em typed path (on $\fr Q$)} to be a sequence $((w_i,\Phi_i))_{i < n}$ such that $(w_i)_{i<n}$ is a path, for all $i< n$, $\Phi_i\sqsubT \ell(w_i)$, and for all $i < n - 1$, $\Phi_i \ST \Phi_{i+1}$.
	
We say that $((w_i,\Phi_i))_{i < n}$ is {\em properly typed} if ${\rm sub}(\Phi^+_{i + 1}) \subseteq {\rm sub}(\Phi^+_{i})$ for all $i<n-1$, and {\em terminal} if $\Phi^+_{n-1} = \varnothing$.
\end{definition}

Note that we allow $\Phi_i\sqsubT \ell(w_i)$ and not only $\Phi_i = \ell(w_i)$. This will allow us to use finite paths, as temporal formulas can be `forgotten' once they have been realized.
\begin{definition}

We define the {\em weak limit model $\widehat{\fr Q}$} of $\fr Q$ as follows:
\begin{itemize}

\item Define $|\widehat {\fr Q}|$ to be the set of terminal typed paths on $\fr Q$ together with the empty path, which we denote $\epsilon$.

\item For $\alpha = ((w_i,\Phi_i))_{i<n}$, $\beta= ((v_i,\Psi_i))_{i<m} \in |\widehat {\fr Q}|$, define $\alpha \mathrel{\peq_{\widehat {\fr Q}}} \beta$ if $n\leq m$ and for all $i<n$, $w_i \peq v_i$ and $\Phi_i \peqT \Psi_i$.

\item Define $S_{\widehat {\fr Q}}(((w_i,\Phi_i))_{i<n}) = ((w_{i+1},\Phi_{i+1}))_{i<n-1}$; note that $  S_{\widehat {\fr Q}} (\epsilon) = \epsilon$.

\item If $n>0$, define $\ell_{\widehat{\fr Q}} (((w_i,\Phi_i))_{i<n}) = \Phi_0$. Then, set
$\ell_{\widehat{\fr Q}} (\epsilon)^- = \bigcup_{w\in W} \ell(w)^-$ and
$\ell_{\widehat{\fr Q}} (\epsilon)^+ = \varnothing$.

\end{itemize}
\end{definition}

The structure $\widehat {\fr Q}$ we have just defined is always a deterministic quasimodel, as we show in the following lemmas.

\begin{lemma}\label{LemIsDP}
	If $\fr Q$ is an $\landi$-quasimodel then $\widehat{\fr Q}$ is a dynamic poset.
\end{lemma}

\fullproof{
\begin{proof}
We have to prove the following.

\begin{itemize}[wide]

\item
{\em ${\peq_{\widehat {\fr Q}}}$ is a partial order on $|\widehat{\fr Q}|$:} This follows easily from the fact that $\peq$ and $\peqT$ are both partial orders.

\item {\em ${S_{\widehat{\fr Q}}}$ is a function:} This is clear since $S_{\widehat{\fr Q}}(\alpha)$ is defined by removing the first element of $\alpha$ if it exists, otherwise $S_{\widehat{\fr Q}}(\alpha) = \alpha$, and thus $S_{\widehat{\fr Q}}(\alpha)$ is uniquely defined for all $\alpha \in |\widehat{\fr Q}|$.

\item {\em $S_{\widehat{\fr Q}}$ is monotone:}
Let $\alpha = ((w_i,\Phi_i))_{i<n}$ and $\beta = ((v_i,\Psi_i))_{i<m}$.
If $\alpha \mathrel{\peq_{\widehat {\fr Q}} } \beta$
then $n\leq m$ and for all $i<n$, $w_i \peq v_1$ and $\Phi_i \peqT \Psi_i$.
If $n>0$, then we also have $n-1 \leq m-1$ and for all $i < n-1$, $w_{i+1} \peq v_{i+1}$ and $\Phi_{i+1} \peqT \Psi_{i+1}$, i.e.,
$S_{\widehat{\fr Q}} (\alpha) = ((w_{i+1},\Phi_{i+1}))_{i<n-1} \mathrel{\peq_{\widehat {\fr Q}}} ((v_{i+1},\Psi_{i+1}))_{i<m -1} = S_{\widehat{\fr Q}} (\beta),$
as needed. If $n = 0$ then $\alpha = \epsilon$, so that $S_{\widehat{\fr Q}}(\alpha) = \epsilon$ and clearly $\epsilon \mathrel {\peq_{\widehat {\fr Q}}} S(\beta)$.
\end{itemize}
\end{proof}
}

\shortproof{
\begin{proof}
We have to prove that ${\peq_{\widehat {\fr Q}}}$ is a partial order on $|\widehat{\fr Q}|$, ${S_{\widehat{\fr Q}}}$ is a function and that it is continuous.
We prove only continuity and leave the other properties to the reader.
Let $\alpha = ((w_i,\Phi_i))_{i<n}$ and $\beta = ((v_i,\Psi_i))_{i<m}$.
If $\alpha \mathrel{\peq_{\widehat {\fr Q}} } \beta$
then $n\leq m$ and for all $i<n$, $w_i \peq v_1$ and $\Phi_i \peqT \Psi_i$.
If $n>0$, then we also have $n-1 \leq m-1$ and for all $i < n-1$, $w_{i+1} \peq v_{i+1}$ and $\Phi_{i+1} \peqT \Psi_{i+1}$, i.e.,
$S_{\widehat{\fr Q}} (\alpha) = ((w_{i+1},\Phi_{i+1}))_{i<n-1} \mathrel{\peq_{\widehat {\fr Q}}} ((v_{i+1},\Psi_{i+1}))_{i<m -1} = S_{\widehat{\fr Q}} (\beta),$
as needed. If $n = 0$ then $\alpha = \epsilon$, so that $S_{\widehat{\fr Q}}(\alpha) = \epsilon$ and clearly $\epsilon \mathrel {\peq_{\widehat {\fr Q}}} S(\beta)$.
\end{proof}
}

%

Next, we must show that $\widehat {\fr Q}$ has `enough' paths. First we show that we can iterate the forward-confluence property.

\begin{lemma}\label{LemmPath}
If $\fr Q$ is an $\landi$-quasimodel, $((w_i,\Phi_i))_{i<n}$ is a typed path in $\fr Q$, and $w_0 \peq v_0$, then there is a typed path $((v_i,\Psi_i))_{i<n}$ such that $w_i \peq v_i$ and $\Phi_i \peqT \Psi_i$ for all $i<n$.
\end{lemma}

\begin{proof}
First we find $v_i$ by induction on $i$; $v_0$ is already given, and once we have found $v_i$, we use forward confluence to choose $v_{i+1}$ so that $v_i \mathrel S v_{i+1}$ and $w_{i+1} \peq v_{i+1}$. Then we set $\Psi_i = \ell(v_i)$; since $S$ is sensible, $\Psi_{i} \ST \Psi_{i+1}$, and by Lemma \ref{LemmRestrict}.\ref{ItCrossTrans}, $\Phi_n \peqT \Psi_n$.
\end{proof}

Now we want to prove that any point can be included in a terminal typed path. For this we will first show that we can work mostly with properly typed paths, thanks to the following.

\begin{lemma}\label{LemProperPath}
Let $\fr Q$ be an $\landi$-quasimodel, $(w_i)_{i <n}$ be a path on $|\fr Q|$, and $\Phi_0 \sqsubseteq \ell (w_0)$. Then there exist $(\Phi_i)_{i<n}$ such that $((w_i,\Phi_i))_{i<n}$ is a properly typed path.
\end{lemma}

\begin{proof}
For $i < n-1$ define recursively $\Phi_{i+1} = \posres{ \ell (w_{i + 1}) }{ {\rm sub} ( \Phi^+_i ) }$; by the assumption that $S$ is sensible and Lemma \ref{LemmRestrict}, $\Phi_i \ST \Phi_{i+1}$ for each $i<n-1$. It is easy to see that $((w_i,\Phi_i))_{i<n}$ thus defined is properly typed.
\end{proof}

However, the properly typed paths we have constructed need not be terminal. This will typically require extending them to a long-enough path.
The extension procedure is precisely the crux of our unwinding procedure.

\begin{lemma}\label{LemTerminal}
If $\fr Q$ is an $\landi$-quasimodel, then any non-empty typed path on $\fr Q$ can be extended to a terminal path.
\end{lemma}

\begin{proof}
Let $\alpha = ((w_i,\Phi_i))_{i<m}$ be any typed path on $\fr Q$. For a type $\Phi$, define $\|\Phi\| = |{\rm sub} (\Phi^+)|$.
We proceed to prove the claim by induction on $\|\Phi_{m-1}\|$. Consider first the case where $\Phi_{m-1}^+$ contains no temporal formulas; that is, formulas of the form $\tnext \psi$ or $\diam \psi$ for some $\psi$. In this case, using the seriality of $S$ choose $w_{m}$ such that $w_{m-1} \mathrel S w_{m}$, and define $\Phi_{m+1} = (\ell(w_{m})^-;\varnothing)$; it is easy to see that $((w_i,\Phi_i))_{i\leq m}$ is a terminal path.
Otherwise, let $\varphi$ be a maximal temporal formula of $\Phi^+_{m-1}$, i.e., it does not appear as a proper subformula of any other temporal formula in $\Phi^+_{m-1}$.
We consider two sub-cases.

Assume first that $\varphi = \tnext \psi$. Then, by the seriality of $S$, we may choose $w_m$ so that $w_{m-1} \mathrel S w_{m}$.
Applying Lemma \ref{LemProperPath}, let $\widetilde \Phi_m$ be such that $((w_{m-1},\Phi_{m-1}), (w_m,\widetilde \Phi_m))$ is a properly typed path.
Setting $\Phi_m = \remove {\widetilde \Phi_m }{\tnext\psi}$, we see by Lemma \ref{LemDelete}.\ref{ItDeleteCirc} that
$((w_{m-1},\Phi_{m-1}),(w_m,\Phi_m))$
is a properly typed path, and $\| \Phi_m \| < \|  \Phi_{m-1} \|$, since the left-hand side does not count $\tnext\psi$. Thus we may apply the induction hypothesis to obtain a terminal typed path $((w_i,\Phi_i))_{i<n}$ extending $\alpha$.

Now consider the case where $\varphi = \diam \psi$. Since $S$ is $\omega$-sensible, there is a path
$w_{m-1} \mathrel S w_m \mathrel S \ldots \mathrel S w_k$
so that $\varphi \in \ell( w_k )$.
Using the seriality of $S$, choose $w_{k+1}$ so that $w_k \mathrel S w_{k+1}$.
By Lemma \ref{LemProperPath}, there are types $\Phi_i$ for $m\leq i \leq k$ and a type $\widetilde \Phi_{k+1}$ such that
$((w_{m-1},\Phi_{m-1}), \ldots, (w_k,\Phi_k), (w_{k+1},\widetilde \Phi_{k+1}))$
is a properly typed path. Then, define $\Phi_{k + 1} = \remove{\widetilde \Phi_{k+1} }{\diam \psi}$. Using Lemma \ref{LemDelete} we see that $(\Phi_k,\Phi_{k+1})$ is sensible; moreover, $\|  \Phi_{k+1} \| < \| \Phi_{m-1} \|$. Hence we can apply the induction hypothesis to obtain a terminal typed path $((w_i,\Phi_i))_{i <n}$ extending $\alpha$.
\end{proof}
%

With this, we are ready to show that our unwinding is indeed a deterministic quasimodel.

\begin{lemma}\label{LemIsQuasi}
If $\fr Q$ is an $\landi$-quasimodel, then $\widehat{\fr Q}$ is a deterministic $\landi$-quasimodel.
\end{lemma}

\begin{proof}
We have already seen in Lemma \ref{LemIsDP} that $(|\widehat{\fr Q}|,\peq_{\widehat {\fr Q}}, S_{\widehat{\fr Q}} )$ is a dynamic poset, so it remains to check that $(|\widehat{\fr Q}|,\peq_{\widehat {\fr Q}}, \ell_{\widehat{\fr Q}})$ is a labelled frame and $S_{\widehat{\fr Q}}$ is sensible and $\omega$-sensible. Let $\alpha = \big ( (w_i,\Phi_i) \big ) _{i<n} \in |\widehat{\fr Q}|$.

First we must check that if $\alpha \mathrel{\peq_{\widehat {\fr Q}}} \beta$, then $\ell_{\widehat{\fr Q}}(\alpha) \peqT \ell_{\widehat{\fr Q}}(\beta)$. Consider two cases; if $n>0$, then $\beta$ is also of the form $(v_i,\Psi_i)_{i<m}$ with $m>0$ and by definition,
$\ell_{\widehat{\fr Q}}(\alpha) = \Phi_0 \peqT \Psi_0 = \ell_{\widehat{\fr Q}}( \beta ).$
Otherwise, $\alpha = \epsilon$, and it is clear from the definition of $\ell_{\widehat{\fr Q}}(\epsilon )$ that $\ell_{\widehat{\fr Q}}(\epsilon ) \peqT \ell_{\widehat{\fr Q}}( \beta )$ regardless of $\beta$.

Now assume that $\varphi \to \psi \in \ell_{\widehat{\fr Q}}(\alpha)^-$. If $n>0$, then since $\fr Q$ is a labelled frame, we can pick $v_0 \seq w_0$ with $ \varphi \in \ell(v_0)^+$ and $\psi \in \ell(v_0)^-$. Since $\Phi_0 \sqsubT \ell(w_0) \peqT \ell(v_0)$, by Lemma \ref{LemmRestrict}.\ref{ItCrossTrans}, $\Phi_0 \peqT \ell(v_0)$, so that by Lemma \ref{LemmPath}, there is a typed path $\beta' = ((v_i,\Psi_i))_{i<n}$ with $\Psi_0 = \ell(v_0)$ such that $w_i \peq v_i$ and $\Phi_i \peqT \Psi_i$ for all $i<n$. By Lemma \ref{LemTerminal}, we can extend $\beta'$ to a terminal path $\beta$. Then, it is easy to see that $\alpha \mathrel{\widehat\peq} \beta$, $ \varphi \in \ell_{\widehat{\fr Q}}( \beta )^+$, and $\psi \in \ell_{\widehat{\fr Q}}( \beta )^-$, as required.

To check that $S_{\widehat{\fr Q}}$ is sensible, consider two cases. If $S_{\widehat{\fr Q}}(\alpha) \not = \epsilon$, then $\alpha$ has length at least two, but since $\alpha$ is a typed path,
$\ell_{\widehat{\fr Q}}(\alpha) = \Phi_0 \ST \Phi_1 = \ell_{\widehat{\fr Q}} \big ( S_{\widehat{\fr Q}} (\alpha) \big ) .$
Otherwise, $S_{\widehat{\fr Q}}(\alpha) = \epsilon$; this means that either $\alpha = \epsilon$ and thus $\ell_{\widehat{\fr Q}} ^+ (\alpha) = \varnothing$, or $\alpha$ has length $1$, in which case since $\alpha$ is terminal, so we also have that $\ell_{\widehat{\fr Q}} ^+ (\alpha) = \varnothing$. In either case, there can be no temporal formula in $\ell_{\widehat{\fr Q}} ^+ (\alpha)$. Now, if $\tnext\varphi \in \ell_{\widehat{\fr Q}} ^- (\alpha)$, then $\tnext \varphi \in \ell ^- (w )$ for some $w \in W$, hence $\varphi \in \ell ^- S(v)$ for any $v$ with $w \mathrel S v$ (which exists since $S$ is serial), and thus $\varphi \in \ell_{\widehat{\fr Q}} (\epsilon)^-$.
Similarly, if $\diam\varphi \in \ell_{\widehat{\fr Q}} ^- (\alpha)$, then by Definition \ref{def:type}.\ref{cond:type:diam} $ \varphi \in  \ell ^- ( \alpha )$, so that $ \varphi \in \ell_{\widehat{\fr Q}} ^- (\epsilon)$.

Finally we check that $S_{\widehat{\fr Q}}$ is $\omega$-sensible. Suppose that $\diam \varphi \in \ell (\alpha)^+$.
This means that $\alpha \not = \epsilon$, so $n > 0$, and since $\alpha$ is terminal, $\diam \alpha \not \in \Phi_{n-1}$.
But this is only possible if $\varphi \in \Phi_i$ for some $i<n-1$, in which case $\varphi \in \ell_{\widehat{\fr Q}} \big ( S_{\widehat{\fr Q}} ^i (\alpha) \big ) $.
\end{proof}

\begin{theorem}\label{TheoConservativity}
A formula $\varphi \in \landi$ is satisfiable (falsifiable) over the class of dynamic posets if and only if it is satisfiable (falsifiable) over the class of dynamical systems.
\end{theorem}

\proof
Suppose that $\varphi \in \mathcal L_\diam$ is satisfied (falsified) on a dynamical topological model. Then, by Theorem \ref{TheoITLc}, $\varphi$ is satisfied (falsified) on some point $w_\ast$ of a ${\rm sub}(\varphi)$-quasimodel $\fr Q $. By Lemma \ref{LemIsQuasi}, $\widehat {\fr Q}$ is a deterministic quasimodel, and by Lemma \ref{LemTerminal}, $(w_\ast, \ell(w_\ast))$ can be extended to a terminal path $\alpha_\ast \in |\widehat{\fr Q}|$. By Lemma \ref{LemTruth}, $\alpha_\ast $ satisfies (falsifies) $\varphi$ on the dynamic poset model $(|\widehat{\fr Q}|,\peq_{\widehat {\fr Q}}, S_{\widehat {\fr Q}}, \val\cdot_{\widehat {\fr Q}})$.
\endproof

As an immediate consequence of Theorems \ref{theocomp} and \ref{TheoConservativity}, we conclude that ${\logbasic} $ is complete for the class of expanding posets.

\begin{corollary}\label{corKripkeComplete}
${\logbasic}  = \itlc_\diam = \itle_\diam$.
\end{corollary}

\section{Completeness for specific spaces}\label{secMetric}

In this section we will show that the above completeness theorems already hold for some familiar spaces, which follows from known results regarding dynamic topological logic.
Thus it will be convenient to briefly review $\sf DTL$ and how $\sf ITL$ embeds into it.
Since the base logic of $\sf DTL$ is classical, we may use a simpler syntax, using the language $\lanclass$ given by the grammar
\[ \bot \  | \   p  \ |  \ \varphi\to\psi  \ |  \ \blacksquare \varphi \  | \  \nx\varphi \  |  \ \nec\varphi \  |  \ \forall \varphi. \]
We can then define $\neg,\wedge,\vee,\blacklozenge,\ps$ using standard classical validities, and denote the $\forall$-free sublanguage by $\lanclass_\nec$.

Given a dynamical system $\mathfrak X$, a {\em classical valuation on $\mathfrak X$} is a function $\cval\cdot\colon\lanclass\to \mathcal P(|\mathfrak X|)$ such that

\[
\begin{array}{lcllcl}
\cval \bot&=&\varnothing &
\cval{\varphi\to\psi}&=& (|\mathfrak X|\setminus\cval\varphi)\cup \cval\psi\\[.8ex]
\cval{\blacksquare\varphi}&=&(\cval\varphi)^\circ&
\cval{\nx\varphi}&=&f^{-1}_\mathfrak X\cval\varphi\\[.8ex]
\cval{\nec\varphi}&=&\displaystyle\bigcap_{n<\omega}f^{-n}_\mathfrak X\cval\varphi\phantom{aaaaa}&
\cval{\forall\varphi}&=&
|\mathfrak X|\text{ if $\cval\varphi=|\mathfrak X|$, else
$\varnothing$.}\\
\end{array}
\]

Note that valuations of formulas are no longer restricted to open sets. We then have the following results regarding satisfiability on some standard metric spaces.

\begin{theorem}[\cite{FernandezMetric}]\label{TheoQComplete}
If $\varphi \in \lanclass$ is classicaly satisfiable on any topological space, then it is classically satisfiable on $\mathbb Q$.
\end{theorem}

\begin{theorem}[\cite{FernandezR2}]\label{theoR2Complete}
If $\varphi \in \lanclass_\nec$ is classicaly satisfiable on an expanding poset, then it is classically satisfiable on $\mathbb R^n$ for any $n\geq 2$.
\end{theorem}

\begin{theorem}[\cite{FernandezMetric}]\label{theoCantorComplete}
If $\varphi \in \lanclass_\nec$ is classicaly satisfiable on any complete metric space, then it is classically satisfiable on the Cantor space.
\end{theorem}

\begin{remark}
Fern\'andez-Duque \cite{FernandezMetric} states Theorem \ref{TheoQComplete} for $\lanclass_\nec$, but the proof provided yields the result for all of $\lanclass$.
Roughly speaking, this is due to the fact that the simulations constructed in the proof are total and surjective.

Theorem \ref{theoR2Complete} is a strengthening of a result of Slavnov \cite{slav}.
\end{remark}

Our intuitionistic temporal logic may then be interpreted in $\sf DTL$ via the G\"odel-Tarski translation $\cdot^\blacksquare$, defined as follows:

\begin{definition}
Given $\varphi\in \lanfull$, we define $\varphi^\blacksquare\in \lanclass$ recursively by setting
\[
\begin{array}{lclclclclcl}
\bot^\nb & = & \bot & \ \ & (\nec \varphi)^\nb & = &\nb \nec \varphi^\nb  & \ \ & (\varphi \odot  \psi )^\nb & = & \varphi^\nb \odot \psi^ \nb \phantom{aaaaa}\\
p^\nb & = & \nb p &  & (\boxdot  \varphi)^\nb & = & \boxdot \varphi^\nb &  & (\varphi\imp \psi )^\nb & = &\nb ( \varphi^\nb \to \psi^ \nb )\\
\end{array}
\]
where $\odot\in \{\wedge,\vee\}$ and $\boxdot \in \{\nx,\ps, \forall\}$.
\end{definition}

In words, we put $\nb$ in front of variables, implication and $\nec$.
The following can then be verified by a simple induction on $\varphi$:

\begin{lemma}\label{LemGT}
Let $\varphi\in \lanfull$, and $\mathfrak X$ be any dynamic topological system. Suppose that $\val\cdot$ is an intuitionistic valuation and $\cval\cdot$ a classical valuation such that, for every atom $p$, $(\cval{p})^\circ=\val p$. Then, for every formula $\varphi$, $\val\varphi=\cval{\varphi^\blacksquare}$.
\end{lemma}

Then we obtain the following.

\begin{theorem}
Given $n\geq 2$, ${\sf ITL}^0_\ps$ is complete for $\mathbb R^n$, as well as for the Cantor space.
\end{theorem}

\proof
If $\varphi $ is valid on $\mathbb R^n$ then by Lemma \ref{LemGT} $\varphi ^\nb$ is valid on $\mathbb R^n$ and hence, by Corollary \ref{corKripkeComplete}, $\varphi ^\nb$ (and thus $\varphi$) is valid on the class of expanding posets and therefore derivable in ${\sf ITL}^0_\ps$.
Completeness for the Cantor space follows by similar reasoning using Theorem \ref{theoCantorComplete}.
\endproof

By similar reasoning, but using Theorems \ref{TheoQComplete} and \ref{theoUnivComp}, we obtain the following.

\begin{theorem}
${\sf ITL}^0_{\ps \forall}$ is complete for $\mathbb Q$.
\end{theorem}

In contrast Example \ref{examRInco} shows that ${\sf ITL}^0_\nx$ is already incomplete for $\mathbb R$.
We leave the problem of axiomatizing ${\sf ITL}^\mathbb R_\nx$ open, along with the long-standing problem of axiomatizing ${\sf DTL}^\mathbb R_\nx$.

\section{Concluding remarks}\label{secConc}

We have provided a sound and complete axiomatization for $\ubox$-free fragments of intuitionistic temporal logics interpreted over various classes of dynamical systems.
Many questions remain open in this direction, perhaps most notably an extension to the full language with $\ubox$. This is likely to be a much more challenging problem than that for $\diam$, as we do not even have a feasible set of axioms for Kremer's interpretation of $\ubox$.
On the other hand, the semantics for $\ubox$ given in \cite{BoudouLICS} does satisfy the standard axioms for $\nec$ of classical $\sf LTL$, but little else is known about it, including its decidability.
Aside from $\sf ITL$ with `henceforth' over the class of all dynamical systems one may consider the corresponding logics for the class of dynamical posets, for spaces with an open map, or for persistent posets; none of these logics have been axiomatized, but we know that they are all distinct \cite{BoudouLICS}.


\bibliographystyle{plain}

\end{document}